\documentclass[10pt,onecolumn,twoside]{IEEEtran}

\IEEEoverridecommandlockouts                              

\usepackage{amsthm}
\usepackage{graphics} 
\usepackage{epsfig} 
\usepackage{times} 
\usepackage{amsmath} 
\usepackage{amssymb}  
\usepackage{adjustbox}
\usepackage{biblatex}
\usepackage{bm}
\usepackage{xcolor}
\usepackage{graphicx}
\usepackage{algorithm}
\usepackage[noend]{algpseudocode}
\usepackage{booktabs}
\usepackage{siunitx}
\usepackage{hyperref}

\usepackage{caption}
\usepackage{subcaption} 
\usepackage{lipsum}
\usepackage{multirow}
\usepackage{diagbox} 
\usepackage{subcaption}

\newcommand{\sigi}{\sigma_i}
\newcommand{\sigj}{\sigma_j}

\newcommand{\Ei}{\mathbb{E}_{\mathbb{P} \mid {\mathcal{F}^i}}}

\newtheorem{definition}{\bf Definition}
\newtheorem{assumption}{\bf Assumption}
\newtheorem{theorem}{\bf Theorem}
\newtheorem{lemma}{\bf Lemma}
\newtheorem{corollary}{\bf Corollary}
\newtheorem{proposition}{\bf Proposition}
\newtheorem{remark}{\bf Remark}

\title{\LARGE \bf
Distributionally Robust Cascading Risk in Multi-Agent Rendezvous: Extended Analysis of Parameter-Induced Ambiguity

}

\author{Vivek Pandey and Nader Motee 
\thanks{
 V. Pandey and N. Motee are with the Department of Mechanical Engineering and Mechanics, Lehigh University, Bethlehem, PA, 18015, USA. {\tt\small \{vkp219,  motee\}@lehigh.edu}.\endgraf
}
}

\begin{document}

\maketitle

\thispagestyle{plain}
\pagestyle{plain}

\begin{abstract}
Ensuring safety in autonomous multi-agent systems during time-critical tasks such as rendezvous is a fundamental challenge, particularly under communication delays and uncertainty in system parameters. In this paper, we develop a theoretical framework to analyze the \emph{distributionally robust risk of cascading failures} in multi-agent rendezvous, where system parameters lie within bounded uncertainty sets around nominal values. Using a time-delayed dynamical network as a benchmark model, we quantify how small deviations in these parameters impact collective safety. We introduce a \emph{conditional distributionally robust functional}, grounded in a bivariate Gaussian model, to characterize risk propagation between agents. This yields a \emph{closed-form risk expression} that captures the complex interaction between time delays, network structure, noise statistics, and failure modes. These expressions expose key sensitivity patterns and provide actionable insight for the design of robust and resilient multi-agent networks. Extensive simulations validate the theoretical results and demonstrate the effectiveness of our framework.
\end{abstract}


\section{Introduction}
Networked systems are pervasive across a wide range of engineering applications, including power grids~\cite{dorfler2012synchronization} and autonomous vehicle platoons~\cite{liu2021risk}. Despite their utility, these systems are often vulnerable to communication delays and external disturbances, which can drive the system away from its desired operating state. Such deviations may lead to significant inefficiencies or, in extreme cases, total system failure. A disruption at a single node or link in a networked system can propagate through the network via interconnections, amplifying its impact and triggering secondary failures. Such cascading effects are commonly observed in critical infrastructures such as power networks~\cite{Somarakis2020power}, supply chains~\cite{bertsimas1998air}, and financial systems~\cite{acemoglu2015systemic}, where local disturbances can lead to widespread systemic breakdowns.

The inherent complexity of networked systems often brings increased fragility, making them vulnerable to systemic failures. In consensus networks~\cite{olfati2004consensus}, where multiple agents work together to reach a unified decision, disruptions such as time delays and environmental noise can lead to deviations from consensus. These sources of uncertainty have spurred growing interest in the risk analysis of complex dynamical systems~\cite{somarakis2018risk,Somarakis2016g,Somarakis2019g}. Accurately understanding and quantifying such risks is critical for designing robust systems, as the failure of a single agent can cascade through the network, potentially undermining overall stability~\cite{zhang2018cascading,zhang2019robustness}.

In the realm of multi-agent consensus networks, various risk assessment methods—such as Value at Risk (VaR)~\cite{rockafellar2000optimization} and Conditional Value at Risk (CVaR)~\cite{rockafellar2002conditional}—have been employed to quantify the risk of both single-agent failures~\cite{Somarakis2016g, Somarakis2017a, Somarakis2019g} and cascading failures~\cite{liu2021risk,liu2023cr_risk_first_order}. These approaches typically rely on the assumption that the underlying probability distributions of uncertainties are known and fixed. However, in practical settings, system parameters such as time delays, diffusion coefficients, and network topology can introduce significant variability and uncertainty, resulting in "distributional ambiguity"—a scenario where the true probability distribution is either unknown or subject to change over time. This ambiguity challenges the validity of traditional risk measures and highlights the need for alternative evaluation frameworks that do not depend on precise distributional knowledge.


To address this challenge, we propose a distributionally robust risk framework that explicitly incorporates ambiguity arising from known, bounded variations in system parameters such as time delays, diffusion coefficients, and network topology. Rather than assuming a single, fixed probability distribution, we define an \emph{ambiguity set}—a family of plausible distributions that are consistent with available data or prior structural knowledge. This formulation enables the evaluation of risk measures under worst-case distributional scenarios within the ambiguity set, thereby ensuring robustness in risk-aware decision-making for networked systems. By quantifying the worst-case risk, our approach provides a conservative yet principled basis for designing systems resilient to distributional uncertainty. Our approach draws inspiration from robust optimization in finance, where worst case risk is quantified under covariance uncertainty \cite{boyd2004convex}.

\subsection*{Our Contributions}
Building on the distributionally robust risk framework developed in our previous work on platoons of vehicles \cite{pandey2023dr_risk_second_order}, this paper advances the theory and application of risk-aware design in multi-agent systems by addressing the multi-agent rendezvous problem under uncertainty. Our key contributions are:
\begin{itemize}
    \item We formulate a conditional distributionally robust functional to define risk in a multi-agent rendezvous setting and quantify this risk using the steady-state statistics of a time-delayed linear consensus network.
    
    \item We derive an explicit, closed-form formula for the distributionally robust cascading risk—capturing the risk of large deviations under the condition that one agent has already failed.
    
    \item We show through theoretical analysis and simulations that higher network connectivity does not necessarily reduce risk, contrary to common intuition.
    
    \item We develop new upper and lower bounds for the cascading risk as functions of network parameters such as edge weights and time delays, providing design guidelines for robust networked systems.
    
    \item Our theoretical results are validated via extensive simulations that illustrate how cascading risk evolves under different network topologies and parameter choices.
\end{itemize}

\subsection*{Relation to Prior Work}

This paper significantly extends our earlier conference paper~\cite{pandey2025distributionallyrobustcascadingrisk}. 
In particular, we broaden the ambiguity set from capturing uncertainty solely in the diffusion coefficient to incorporating uncertainty in time delays and edge weights as well. 
We further reformulate the DR risk problem—originally posed as a supremum over the ambiguity set—into a tractable optimization problem. 
In addition, we derive new analytical bounds that relate the DR risk explicitly to network connectivity and time-delay parameters, thereby offering deeper insight into system robustness.

\section{Preliminaries and Notation}
We denote the non-negative orthant of the Euclidean space \(\mathbb{R}^n\) by \(\mathbb{R}_{+}^n\). The standard Euclidean basis vectors in \(\mathbb{R}^n\) are given by \(\{\bm{e}_1, \dots, \bm{e}_n\}\), and the all-ones vector is denoted by \(\bm{1}_n = [1, \dots, 1]^\top\). The set of \(n \times n\) positive semidefinite matrices is denoted by \(\mathcal{S}_+^n\). For any two matrices \(A, B \in \mathbb{R}^{n \times n}\), we write \(A \succeq B\) to indicate that \(A - B\) is positive semidefinite.

\vspace{0.1cm}
{\it Spectral Graph Theory:} A weighted graph is characterized by \(\mathcal{G} = (\mathcal{V}, \mathcal{E}, \omega)\), where \(\mathcal{V}\) is the set of nodes, \(\mathcal{E}\) is the set of edges, and \(\omega: \mathcal{V} \times \mathcal{V} \to \mathbb{R}_{+}\) is the weight function that assigns a non-negative value to each edge. Two nodes are said to be directly connected if and only if \((i, j) \in \mathcal{E}\).

\begin{assumption} \label{asp:connected}
    The graph under consideration is assumed to be connected. Furthermore, for all \(i, j \in \mathcal{V}\), the following conditions hold:
    \begin{itemize}
        \item \(\omega(i,j) > 0\) if and only if \((i,j) \in \mathcal{E}\).
        \item \(\omega(i,j) = \omega(j,i)\), i.e., the edges are undirected.
        \item \(\omega(i,i) = 0\), implying that the edges are simple.
    \end{itemize}
\end{assumption}

The Laplacian matrix of the graph \(\mathcal{G}\) is an \(n \times n\) matrix \(L = [l_{ij}]\) with elements given by
\[
    l_{ij} := 
    \begin{cases}
        -k_{ij}  & \text{if } i \neq j, \\
        k_{i1} + \dots + k_{in}  & \text{if } i = j,
    \end{cases}
\]
where \(k_{ij} := \omega(i,j)\). The Laplacian matrix is symmetric and positive semi-definite \cite{van2010graph}. Assumption~\ref{asp:connected} implies that the smallest eigenvalue of the Laplacian is zero, with an algebraic multiplicity of one. The eigenvalues of \(L\) can be ordered as
\[
    0 = \lambda_1 < \lambda_2 \leq \dots \leq \lambda_n.
\]
The eigenvector corresponding to the eigenvalue \(\lambda_k\) is denoted by \(\bm{q}_k\). By letting \(Q = [\bm{q}_1 \; | \; \dots \; | \; \bm{q}_n]\), we can express the Laplacian matrix as \(L = Q \Lambda Q^T\), where \(\Lambda = \text{diag}(0, \lambda_2, \dots, \lambda_n)\). We normalize the Laplacian eigenvectors so that \(Q\) becomes an orthogonal matrix, i.e., \(Q^T Q = Q Q^T = I_n\), with \(\bm{q}_1 = \frac{1}{\sqrt{n}} \bm{1}_n\).

\vspace{0.1cm}

{\it Probability Theory:} Let \(\mathcal{L}^2(\mathbb{R}^q)\) denote the space of \(\mathbb{R}^q\)-valued random vectors \(\bm{z} = [z^{(1)}, \dots, z^{(q)}]^T\) defined on a probability space \((\Omega, \mathcal{F}, \mathbb{P})\) with finite second moments. A normally distributed random vector \(\bm{y} \in \mathbb{R}^q\) with mean \(\bm{\mu} \in \mathbb{R}^q\) and covariance matrix \(\Sigma \in \mathbb{R}^{q \times q}\) is denoted by \(\bm{y} \sim \mathcal{N}(\bm{\mu}, \Sigma)\). The error function \(\operatorname{erf} : \mathbb{R} \to (-1, 1)\) is defined as
\[
\operatorname{erf}(x) = \frac{2}{\sqrt{\pi}} \int_0^x e^{-t^2} \, \mathrm{d}t.
\]
We also adopt the standard notation \(\mathrm{d}\bm{\xi}_t\) in the context of stochastic differential equations.

\section{Problem Formulation}\label{problemstatement}

We focus on time-delayed linear consensus networks, which have broad applications in engineering, including clock synchronization in sensor networks, rendezvous in space or time, and heading alignment in swarm robotics. For more details, we refer to \cite{ren2007information, olfati2007consensus}. As a motivating example, we explore the time-delayed rendezvous problem, where the goal is for a group of agents to meet simultaneously at a pre-specified location known to all.

In this scenario, agents lack prior knowledge of the meeting time, which may need to be adjusted in response to unforeseen emergencies or external uncertainties \cite{ren2007information}. Thus, the agents must reach a consensus on the rendezvous time. This consensus is achieved by each agent \( i = 1, \dots, n \) creating a state variable, say \( x_i \in \mathbb{R} \), which represents its belief about the rendezvous time. Initially, each agent sets its belief to the time it prefers for the rendezvous. The dynamics of each agent's belief evolve over time according to the following stochastic differential equation:
\begin{equation} \label{eqn:individual_dynamics}
\text{d} x_i(t) = u_i(t) \, \text{d} t + b \, \text{d} w_i(t),
\end{equation}
for all \( i = 1, \dots, n \).

Each agent's control input is denoted by \( u_i \in \mathbb{R} \). Uncertainty in the network propagates as additive stochastic noise, with its magnitude uniformly scaled by the diffusion coefficient \( b \in \mathbb{R} \). The influence of environmental uncertainties on agent dynamics is modeled using independent Brownian motions \( w_1, \dots, w_n \).

In practical scenarios, agents often experience time delays in accessing, computing, or sharing their own and neighboring agents' state information~\cite{ren2007information}. To account for this, we assume that all agents experience a uniform time delay \( \tau \in \mathbb{R}_{+} \).

The control inputs are determined through a negotiation process, in which agents interact over a communication graph to form a linear consensus network. The control law governing the system is given by
\begin{equation} \label{eq:feedback}
    u_i(t) = \sum_{j = 1}^{n} k_{ij} \left(x_j(t-\tau) - x_i(t-\tau) \right),
\end{equation}
where \( k_{ij} \in \mathbb{R}_{+} \) are nonnegative feedback gains.

Let us denote the state vector as \( \bm{x}_t = \left[x_1(t), \dots, x_n(t) \right]^T \), and the vector of exogenous disturbances as \( \bm{w}_t = \left[w_1(t), \dots, w_n(t) \right]^T \). The dynamics of the resulting closed-loop network can then be expressed as a linear consensus network governed by the following stochastic differential equation:
\begin{equation} \label{eqn: network_dynamics}
    \text{d} \bm{x}_t = -L \, \bm{x}_{t-\tau}\, \text{d}t + B \, \text{d} \bm{w}_t,
\end{equation}
for all \( t \geq 0 \), where the initial function \( \bm{x}_t = \phi(t) \) is deterministically specified for \( t \in [-\tau, 0] \), and \( B = b I_n \).

The underlying coupling structure of the consensus network in \eqref{eqn: network_dynamics} is represented by a graph \( \mathcal{G} \) that satisfies Assumption~\ref{asp:connected}, with the corresponding Laplacian matrix \( L \). We assume that the communication graph \( \mathcal{G} \) is time-invariant, ensuring that the network of agents reaches consensus on a rendezvous time before executing motion planning to the designated meeting location. Once consensus is established, a properly designed internal feedback control mechanism drives each agent toward the rendezvous point.

\begin{assumption} \label{asp:stable}
\cite{Somarakis2019g} The time delay is assumed to satisfy
\[
\tau < \frac{\pi}{2 \lambda_n},
\]
which guarantees the stability of the network dynamics \eqref{eqn: network_dynamics}.
\end{assumption}

In the absence of noise, i.e., \( b = 0 \), it is known from \cite{olfati2004consensus} that under Assumptions~\ref{asp:connected} and \ref{asp:stable}, the state of each agent converges to the average of the initial conditions, namely,
\[
\frac{1}{n}\bm{1}_n^T \bm{x}_0.
\]
In contrast, when input noise is present, the state variables fluctuate around the instantaneous average \( \frac{1}{n}\bm{1}_n^T \bm{x}_t \).

To quantify the quality of the rendezvous and to capture its fragility characteristics, we define the vector of observables as
\begin{equation} \label{eq:observables}
    \bm{y}_t = M_n\, \bm{x}_t,
\end{equation}
where the centering matrix \( M_n = I_n - \frac{1}{n}\bm{1}_n \bm{1}_n^T \) ensures that \( \bm{y}_t = \bigl(y_1(t), \dots, y_n(t)\bigr)^T \) captures the deviation of each agent from the consensus value.

Under Assumption~\ref{asp:connected}, the network dynamics described by \eqref{eqn: network_dynamics} include a marginally stable mode associated with the zero eigenvalue of the Laplacian matrix \(L\). This mode corresponds to uniform translations of the state vector and is unobservable from the output \eqref{eq:observables}, ensuring that the observable vector \(\bm{y}_t\) remains bounded in steady state. In the absence of noise, we have \(\bm{y}_t \to \bm{0}\) as \(t \to \infty\), indicating perfect consensus. However, in the presence of noise, the observable modes are persistently excited, causing \(\bm{y}_t\) to fluctuate around zero. As a result, exact agreement on the rendezvous time is no longer possible, and a practical solution is to introduce a tolerance interval within which consensus is considered satisfactory.

\begin{definition} \label{def:c-Consensus} [\textbf{\emph{\(c\)-Consensus}}]
For a given tolerance level \( c \in \mathbb{R}_+ \), the consensus network \eqref{eqn: network_dynamics} is said to achieve \emph{\(c\)-consensus} if it can tolerate a bounded level of disagreement such that
\[
\lim_{t \to \infty} \mathbb{E}[|\bm{y}_t|] \leq c \mathbf{1}_n.
\]
\end{definition}

This condition implies that all agents asymptotically agree within a tolerance \(c\), i.e., the state vector lies in the set
\[
\left\{ x \in \mathbb{R}^n \,\middle|\, \mathbb{E}[|M_n x|] \leq c \mathbf{1}_n \right\}.
\]
However, stochastic disturbances and parameter uncertainty may cause the expected deviation of an agent’s observable from consensus to exceed the tolerance threshold, suggesting a systemic risk.

\begin{definition}[Systemic Failure]\label{def:systemic_failure}
Given a tolerance level \( c \in \mathbb{R}_+ \), agent \( i \) is said to be \emph{at risk of systemic failure} if
\[
\lim_{t \to \infty} \mathbb{E}[|y_t^{(i)}|] > c.
\]
An \emph{actual systemic failure} occurs when it is known that
\[
\lim_{t \to \infty} |y_t^{(i)}| > c.
\]
\end{definition}

\textbf{Problem Statement:} The central challenge addressed in this work is to quantify the \emph{distributionally robust cascading risk} in the multi-agent rendezvous problem under parameter uncertainty. Specifically, we seek to characterize how the interplay between the graph Laplacian, time delay, and noise statistics influences the likelihood of failing to achieve \(c\)-consensus, particularly when some agent has already experienced systemic failure. To this end, we propose a systemic risk framework that evaluates the propagation of large fluctuations through the network by analyzing the steady-state statistics of the closed-loop dynamics. This framework enables a principled assessment of cascading failures in multi-agent systems under distributional ambiguity.

The remainder of the paper is organized as follows. Section~\ref{sec:prelims} presents preliminary results on the steady-state statistics of the system and introduces the distributionally robust risk framework. In Section~\ref{sec:quant_ambiguity_set}, we characterize the ambiguity set of relevant observables arising from uncertainties in diffusion coefficients, time delays, and network edge weights. Section~\ref{sec: dr_risk} derives a closed-form expression for the distributionally robust risk. Section~\ref{sec:fundamental_limits} establishes fundamental bounds on the risk as a function of key network parameters. Section~\ref{sec:case_study} provides simulation-based validation of the theoretical findings. Finally, Section~\ref{sec:conclusion} concludes the paper and outlines directions for future research. The proofs of all theoretical results are provided in Appendix.

\section{Preliminary Results}\label{sec:prelims}

In this section, we analyze the steady-state statistics of network observables in the presence of external disturbances and time delays. We introduce the notion of \emph{systemic events} to capture significant deviations of observables from their mean behavior. To assess the risk of cascading failures under parameter uncertainty, we develop a distributionally robust risk framework that accounts for ambiguity in the underlying probability distribution.

\subsection{Steady-State Statistics of Observables}

It is known from \cite{Somarakis2019g} that under Assumptions~\ref{asp:connected} and \ref{asp:stable}, the steady-state observables defined in \eqref{eq:observables} converge in distribution to a multivariate normal distribution as \( t \to \infty \):
\begin{equation}\label{eqn:steady_state_distribution_observables}
    \bar{\bm{y}} \sim \mathcal{N}(0, \Sigma),
\end{equation}

where \( \Sigma \in \mathbb{S}^n_+ \) is the steady-state covariance matrix.

\begin{lemma} \label{lem:sigma_y_steady}
The steady-state covariance matrix \( \Sigma \) of the observables \eqref{eq:observables} is given by
\begin{equation} \label{eq:sigma_y}
    \Sigma =  b^2 M_n Q \Psi Q^\top M_n,
\end{equation}
where \( \Psi =  \operatorname{diag} \left(0, \tau f(\lambda_2 \tau), \dots, \tau f(\lambda_n \tau) \right) \) is a diagonal matrix, and  \(f(\lambda_i \tau) = \frac{\cos(\lambda_i \tau)}{2\lambda_i \tau \left(1 - \sin(\lambda_i \tau)\right)}\) for all \(i = \{2, \dots, n\}\).\\
For convenience, we denote the \((i,j)\)-th element of \( \Sigma \) as \( \sigma_{ij} \), and the diagonal elements as \( \sigma_{ii} = \sigma_i^2 \).
\end{lemma}

\begin{remark} \label{rem: Sigma_y_time_delay_zero}
In the special case where the time delay \( \tau = 0 \), the covariance matrix \( \Sigma \) simplifies to
\begin{equation} \label{eq:Sigma_y_time_delay_zero}
    \hat{\Sigma} = \frac{1}{2} b^2 L^\dagger,
\end{equation}
where \( L^\dagger \) denotes the Moore–Penrose pseudoinverse of the graph Laplacian \( L \) \cite{van2010graph}.
\end{remark}

Remark~\ref{rem: Sigma_y_time_delay_zero} follows directly by substituting \( \tau = 0 \) into \eqref{eq:sigma_y} and using the identity \( M_n L^\dagger M_n = L^\dagger \).

\vspace{0.5em}

Lemma~\ref{lem:sigma_y_steady} admits an interpretation as an eigen-decomposition of the covariance matrix and reveals the effect of both time delays and network topology on steady-state statistics. This relationship is crucial in quantifying network-induced risks, as explored later in the paper.

To characterize the cascading influence of one random variable on another, it is necessary to understand their joint distributions. We now present the relevant joint distribution model.

\subsection{Joint Normal Distribution}
From \eqref{eq:sigma_y} in Lemma~\ref{lem:sigma_y_steady}, it is evident that the covariance matrix \(\Sigma\) is singular, with its kernel spanned by \(\bm{1}_n\), owing to the fact that \(M_n\) is itself singular with null space spanned by \(\bm{1}_n\). To ensure that the joint density of any two observables \(y_i\) and \(y_j\) is well-defined—that is, their corresponding covariance matrix is non-singular—we first establish the following lemma.

\begin{lemma}\label{lem:principle_covariance_invertibility}
Any principal submatrix (i.e., a matrix obtained by deleting the \(i\)th row and column) of the covariance matrix \(\Sigma\) as defined in \eqref{eq:sigma_y} is invertible.  
\end{lemma}

 The result of Lemma~\ref{lem:principle_covariance_invertibility} plays a crucial and often implicit role in ensuring the invertibility of principal submatrices of the covariance matrix \(\Sigma\). Specifically, Lemma~\ref{lem:principle_covariance_invertibility} guarantees that the steps involved in risk computation are mathematically sound, by ensuring that the \(2 \times 2\) covariance matrix for any pair \( i,j \in \{1, \dots, n\}, i \neq j\) in Lemma~\ref{lem:bivariate_normal} is always non-singular.

Having established the non-singularity of the covariance matrix corresponding to any pair of observables, we now state their joint probability density as follows:

\begin{lemma}\label{lem:bivariate_normal}
Let \( y_i \) and \( y_j \) be jointly normally distributed random variables with correlation coefficient \( \rho = \rho_{ij} \), and standard deviations \( \sigi \) and \( \sigj \) respectively, as defined in Lemma~\ref{lem:sigma_y_steady}. Then, their joint probability density function is given by
\begin{equation}\label{eqn:bivariate_normal_distribution}
    f(y_j, y_i) = \frac{1}{2\pi \rho' \sigi \sigj} \exp\left( 
    -\frac{y_j^2}{2\sigj^2} 
    - \frac{\left( y_i - \rho \frac{\sigi y_j}{\sigj} \right)^2}{2 \rho'^2 \sigi^2} 
    \right),
\end{equation}
where \( \rho' = \sqrt{1 - \rho^2} \).
\end{lemma}

This expression for the joint density is subsequently used to derive an explicit solution for cascading risk, as detailed in later sections.

\subsection{Systemic Events}

A \textit{systemic event} refers to a large deviation in an observable that may trigger cascading failures. The set of such undesirable outcomes is called the \textit{systemic set}, denoted by \( U \subset \mathbb{R} \). 

In a probability space \( (\Omega, \mathcal{F}, \mathbb{P}) \), for a random variable \( y: \Omega \rightarrow \mathbb{R} \), the set of systemic events is defined as
\[
\{ \omega \in \Omega \mid y(\omega) \in U \}.
\]

To monitor how close the observable \( y \) is to entering the systemic set \( U \), we define a family of nested supersets \( \{ U_{\delta} \mid \delta \in [0, \infty) \} \) that satisfy the following properties:
\begin{itemize}
    \item \textit{Nestedness}: \( U_{\delta_2} \subseteq U_{\delta_1} \) whenever \( \delta_1 < \delta_2 \),
    \item \textit{Limit condition}: \( \lim_{\delta \to \infty} U_{\delta} = U \).
\end{itemize}

These supersets \( U_{\delta} \) can be tuned to include a neighborhood around \( U \), effectively forming “alarm zones.” As will be demonstrated in the following sections, the parameter \(\delta\) quantifies the magnitude of deviation of an agent's observable from its mean.

\subsection{Distributionally Robust Risk Framework}  \label{risk}

To quantify risk in the presence of distributional uncertainty, it is important to recognize that the true probability distribution may be unknown or only partially known. To address this challenge, we adopt a framework that considers a family of probability distributions, rather than assuming a single known distribution. This approach enables the evaluation of worst-case scenarios across all plausible distributions, yielding a more resilient and conservative measure of risk under parameter uncertainty.

Central to this approach is the concept of an \emph{ambiguity set}, which characterizes the collection of probability measures considered plausible. This is formalized below.

\begin{definition}\label{def: ambig_set}
\textbf{Ambiguity Set of Probability Measures} \cite{shapiro2022}:  
Let \( (\Omega, \mathcal{F}) \) be a measurable space. An \emph{ambiguity set} \( \mathcal{M} \) is defined as a nonempty family of probability measures on \( (\Omega, \mathcal{F}) \), lying within a specified ball centered at a reference probability measure, with radius determined by a chosen divergence or distance metric.  
\end{definition}

Building on the ambiguity set, we define the \emph{conditional distributionally robust functional}, which forms the foundation of our risk quantification framework. This functional captures the worst-case conditional expectation of a random variable across all measures in the ambiguity set.

\begin{definition}\label{def: condn_drf}
\textbf{Conditional Distributionally Robust Functional} \cite{shapiro2022}:  
Let \( (\Omega, \mathcal{F}) \) be a measurable space, and let \( \mathcal{M} \) be an ambiguity set of probability measures on \( (\Omega, \mathcal{F}) \). For a random variable \( Z: \Omega \to \mathbb{R} \), the \emph{conditional distributionally robust functional} is defined as  
\begin{equation}\label{eq:condn_drf_defn}
    \mathcal{R} :=
    \sup_{\mathbb{P} \in \mathcal{M}} \mathbb{E}_{\mathbb{P}|\mathcal{F}}[Z],
\end{equation}  
where \( \mathbb{E}_{\mathbb{P}|\mathcal{F}}[Z] \) denotes the conditional expectation of \( Z \) with respect to the \( \sigma \)-algebra \( \mathcal{F} \) under measure \( \mathbb{P} \).  
\end{definition}

The functional defined in Definition~\ref{def: condn_drf} serves as the core building block for constructing a distributionally robust risk measure. By incorporating ambiguity in the underlying distribution, this formulation allows us to quantify the risk of a random variable in the worst-case scenario, ensuring robustness to model misspecification.

With this foundation in place, we first quantify the ambiguity set that captures distributional uncertainty, followed by a formal definition and computation of the distributionally robust risk measure in the subsequent section.

\section{Quantifying Distributional Ambiguity}\label{sec:quant_ambiguity_set}

In practical networked systems, parameters such as the diffusion coefficient, time delay, and edge weights are typically estimated or tuned based on empirical observations. As a result, they are subject to uncertainty and may deviate slightly from their nominal values. Even small perturbations in these parameters can lead to non-negligible variations in the system’s behavior and, consequently, in the probability distribution of its observable outputs.

To capture the model uncertainty, we construct an ambiguity set representing a family of probability measures consistent with bounded deviations in the underlying parameters. In particular, uncertainties in the time delay, diffusion coefficient, and network fluctuations—as highlighted in Lemma~\ref{lem:sigma_y_steady}—serve as the primary sources of distributional ambiguity. These combined effects cause the true distribution to deviate from the nominal model, thereby motivating the use of robust analytical methods applied over the constructed ambiguity set.

\begin{remark}\label{rem:ambiguity_set_quant_framework}
Since the network dynamics in \eqref{eqn: network_dynamics} are driven by Brownian motion, the resulting steady-state distribution of the observables \eqref{eqn:steady_state_distribution_observables} is Gaussian with zero mean and a covariance matrix that depends on system parameters, as characterized in Lemma~\ref{lem:sigma_y_steady}. To account for uncertainty in these parameters, we define an \emph{ambiguity set} as a family of zero-mean multivariate Gaussian distributions, where the covariance matrix is bounded in the positive semidefinite order based on the level of parametric uncertainty. This is illustrated in Fig.~\ref{fig:ambiguity_set}.

Specifically, for a parameter \(p \in \{b, \tau, \omega\}\)\footnote{Here, \(b\) is the diffusion coefficient, \(\tau\) the time delay, and \(\omega = \omega_{ij}\) the edge weight of the communication graph. We treat the edge weight as a parameter, not the Laplacian eigenvalues \(\{\lambda_i\}_{i=1}^n\), since edge weights have practical meaning as feedback gains in \eqref{eq:feedback}, while the eigenvalues are functions of these weights.}
, we assume bounded uncertainty around a nominal estimate \(p_0 \in \{b_0,\tau_0,\omega_0\}\), expressed as
\begin{equation}\label{eqn:general_parameter_bound}
    (1 - \alpha_p)p_0 \leq p \leq (1 + \alpha_p)p_0,
\end{equation}
where \(\alpha_p \in [0,1)\) is a known relative uncertainty. \emph{Without loss of generality}, this multiplicative form subsumes additive bounds; for instance, an additive constraint \(|p - p_0| \leq \beta\) can be rewritten as \eqref{eqn:general_parameter_bound} by setting \(\alpha_p = \beta / p_0\), assuming \(p_0 > 0\).

Given such parametric uncertainty, the induced ambiguity set of steady-state distributions is defined as
\begin{equation}\label{eqn:ambiguity_set_quantification}  
    \mathcal{M}_{p} = \left\{ \mathbb{P} \sim \mathcal{N}(0, \Sigma) \,\middle|\, (1-\varepsilon_p^-)\Sigma_0 \preceq \Sigma \preceq (1+\varepsilon_p^+)\Sigma_0 \right\},
\end{equation}
where \(\Sigma_0\)  denotes the nominal covariance obtained by substituting \(p_0\) in \eqref{eq:sigma_y}, \(\varepsilon_p^- \in [0,1)\), and \(\varepsilon_p^{\pm} = \varepsilon^{\pm}(\alpha_p)\) quantifies the radius of the ambiguity set as a function of the relative parameter uncertainty.
\end{remark}

\begin{figure}[t]
\centering
    \includegraphics[width=0.96\columnwidth]{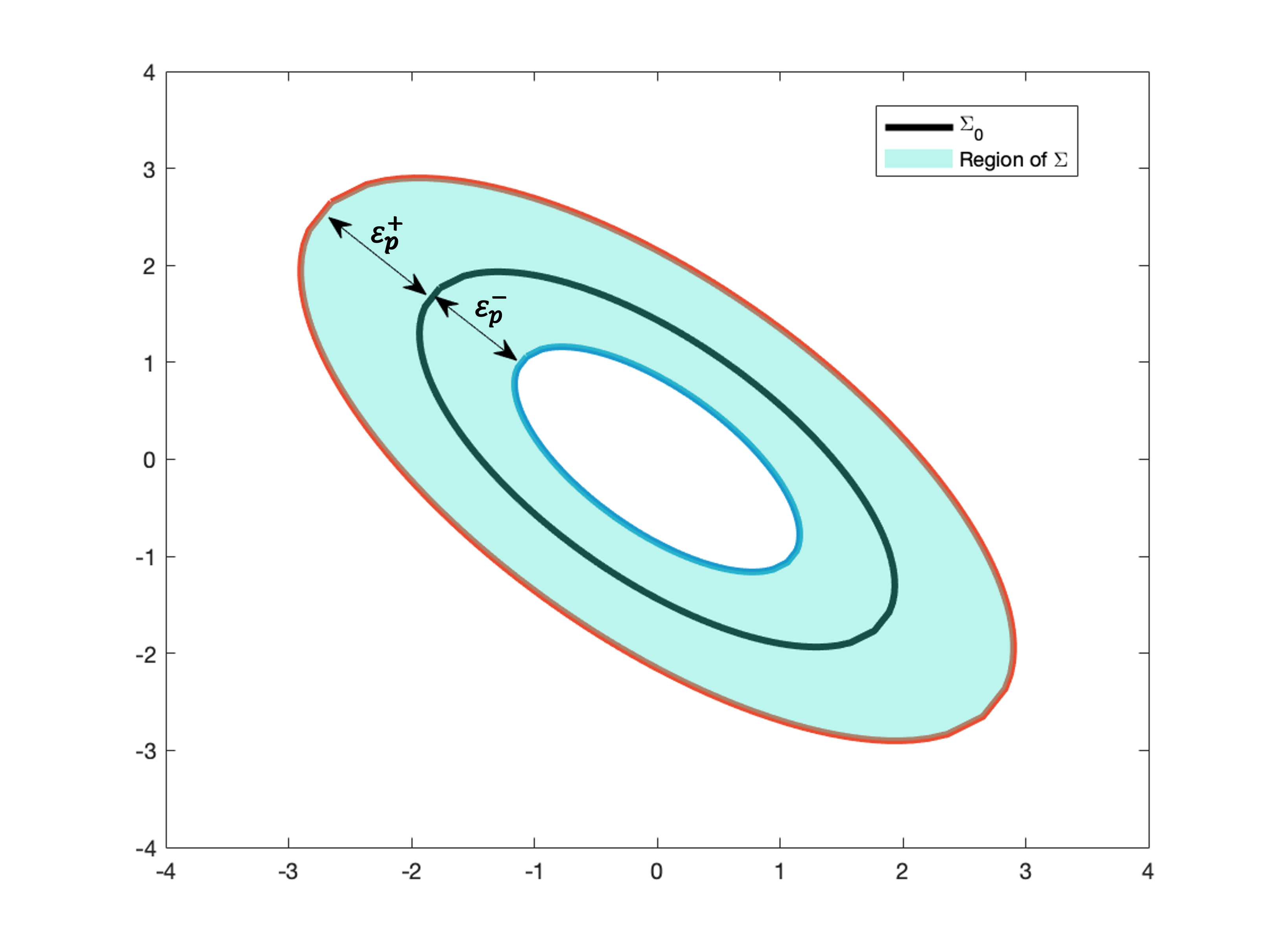} 
    \caption{Geometry of Ambiguity Set  \(\left(\tau \neq 0\right)\). The shaded region illustrates the feasible set of covariance matrices \(\Sigma\) consistent with the assumed distributional ambiguity.}
    \label{fig:ambiguity_set}
\end{figure}

Having established a general framework for quantifying the ambiguity set of the steady-state distribution of observables, we now proceed to examine the contribution of individual sources of uncertainty—namely, the diffusion coefficient, time delay, and edge weights of the communication graph. Specifically, we quantify \(\varepsilon_p\) for each \(p \in \{b, \tau, \omega\}\) independently. While it is indeed possible to construct a unified ambiguity set that captures the joint uncertainty across all parameters, we opt for a parameter-wise decomposition. This approach allows us to isolate and better understand the distinct impact of each type of uncertainty on risk quantification. We begin by analyzing the ambiguity induced by uncertainty in the diffusion coefficient, and subsequently consider the effects of time delay and edge weight variations.

\subsection{Ambiguity due to Diffusion Coefficient}
In practice, the diffusion coefficient is not known exactly; instead, a nominal value \(b_0\) is typically obtained through learning-based methods. However, this estimate inherently carries uncertainty, which is quantified by a parameter \(\alpha_b\).

\begin{proposition}\label{prop:ambiguity_b}
    Consider the system described by \eqref{eqn: network_dynamics}, with nominal diffusion coefficient \(b_0\), and suppose the true coefficient \(b\) satisfies the bounded uncertainty condition
\begin{equation}\label{eqn:ambiguity_set_B}
    (1 - \alpha_b)\, b_0^2 \leq b^2 \leq (1 + \alpha_b)\, b_0^2,
\end{equation}
where \(\alpha_b \in [0,1)\) is a known parameter.

Then, the ambiguity set for the steady-state distribution of the observables \(\bar{\bm{y}}\) is given by
\begin{equation}\label{eqn:ambiguity_B_y}
    \mathcal{M}_{b} = \left\{ \mathbb{P} \sim \mathcal{N}(0, \Sigma) \,\middle|\, (1 - \varepsilon_b)\, \Sigma_0 \preceq \Sigma \preceq (1 + \varepsilon_b)\, \Sigma_0 \right\},
\end{equation}
where \(\Sigma_0\) is the nominal covariance matrix as defined in Remark \ref{rem:ambiguity_set_quant_framework}, and \(\varepsilon_b^{\pm} = \alpha_b\).
\end{proposition}

From Lemma \ref{lem:sigma_y_steady}, it follows that ambiguity set \({\mathcal{M}}_{b}\) can be written as 
        \begin{equation}\label{eqn:ambiguity_set_b}
   {\mathcal{M}}_{b} = \left\{\mathbb{P} \sim \mathcal{N}(0, \Sigma) \mid \lvert \frac{\sigma_{ij} - \sigma_{0,ij}}{\sigma_{0,ij}}\rvert \leq \alpha_b\right\} ,
    \end{equation}
where \(\sigma_{0,ij} =\Sigma_0(i,j).\) and \(\sigma_{ij} =\Sigma(i,j).\)

While this quantification is straightforward under the specific case \( B = b I_n \), as defined in \eqref{eqn: network_dynamics}, we provide a more general treatment for uncertain \( B \) in the Appendix. This generalization, which may be of independent interest, allows for arbitrary \( B \). However, in the main text, we restrict our analysis to the case \( B = b I_n \), since the closed-form expression in \eqref{eq:sigma_y} is valid only under this assumption. We now proceed to ambiguity quantification in the presence of time-delay uncertainty.

\subsection{Ambiguity due to Time Delay}

In many practical scenarios, the exact value of the time delay \(\tau\) is not known. Instead, a data-driven nominal estimate \(\tau_0\) is typically available, accompanied by an uncertainty bound parameter \(\alpha_{\tau}\). Under this modeling assumption, we now quantify the ambiguity set induced by time-delay uncertainty.  

For notational convenience, we define  
\begin{equation}\label{eqn:f_lambda_tau}
    f_{\lambda_i}(\tau) = \frac{ \cos(\lambda_i \tau)}{\lambda_i \left(1 - \sin(\lambda_i \tau)\right)},
\end{equation}
for each \(\lambda_i\), \(i \in \{2, \dots, n\}\). A representative plot of \(f_{\lambda_i}(\tau)\) is shown in Fig.~\ref{fig:f_lambda_tau_graph}, which is monotonic in \(\tau \in \frac{\pi}{2c}\) for a given \(c.\) 

\begin{figure}[t]
    \centering
    \includegraphics[width=0.96\columnwidth]{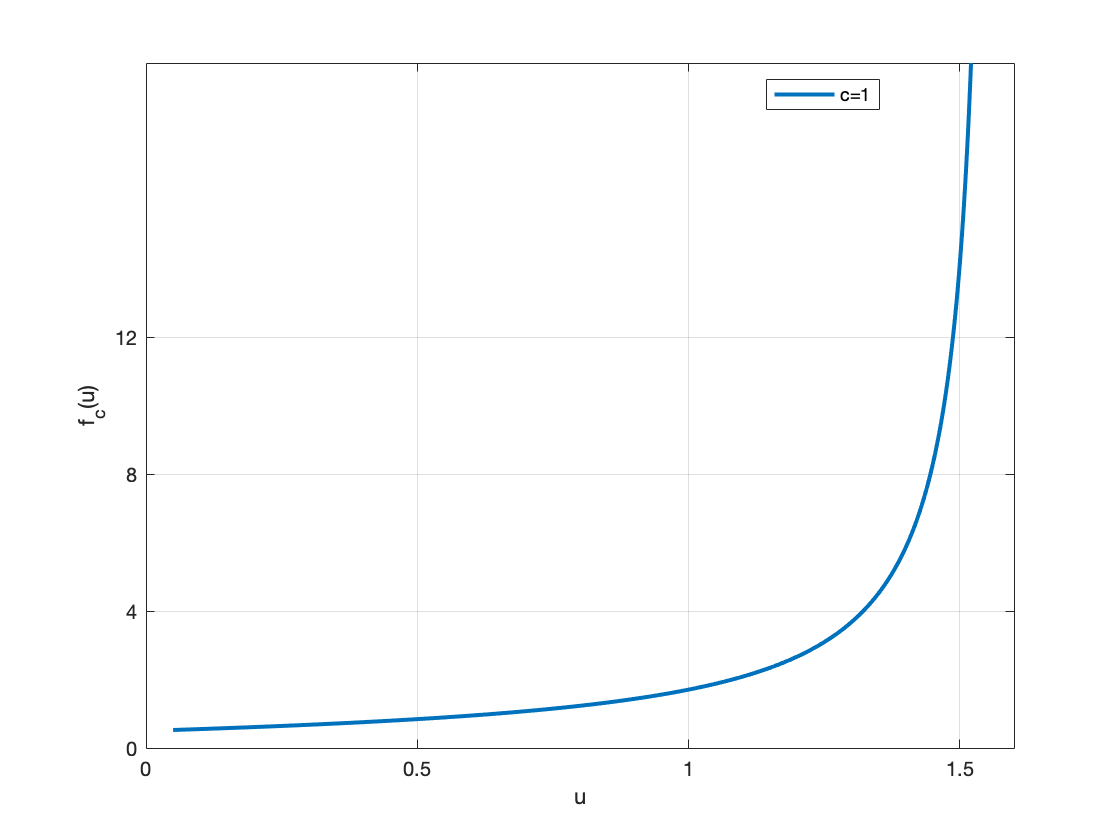} 
    \caption{Graph of \(f_c(u) = \frac{\cos(cu)}{2c (1- \sin(cu))}\).}
    \label{fig:f_lambda_tau_graph}
\end{figure}

\begin{proposition} [Covariance Bounds under Delay Perturbation]\label{prop:ambiguity_time_delay}
Consider the system governed by \eqref{eqn: network_dynamics} with nominal time delay \(\tau_0\) and relative uncertainty \(\alpha_{\tau} \in [0,1)\). Suppose the delay satisfies  
\begin{equation}\label{eqn:time_delay_bound}
   (1 - \alpha_{\tau})\,\tau_0 \ \leq\  \tau \ \leq\  (1 + \alpha_{\tau})\,\tau_0.
\end{equation}
Then, the ambiguity set \(\mathcal{M}_{\tau}\) for the steady-state distribution of the observables \(\bar{\bm{y}}\) is given by  
\begin{equation}\label{eqn:ambiguity_time_delay}
    \mathcal{M}_{\tau} = \left\{ \mathbb{P} \sim \mathcal{N}(0, \Sigma) \ \middle|\ (1 - \varepsilon_{\tau}^-)\,\Sigma_0 \ \preceq\ \Sigma \ \preceq\ (1 + \varepsilon_{\tau}^+)\,\Sigma_0 \right\},
\end{equation}
where  
\[
\varepsilon_{\tau}^{\pm} = \max_{i \in \{2,\dots,n\}} \frac{\left| f_{\lambda_i}\big((1 \pm \alpha_{\tau})\tau_0\big) - f_{\lambda_i}(\tau_0) \right|}{f_{\lambda_i}(\tau_0)},
\]
and \(\Sigma_0\) is the nominal covariance corresponding to \(\tau_0\).
\end{proposition}


Unlike the parameter bounds for the diffusion coefficient, the covariance bounds in this case are \emph{asymmetric} due to the nonlinearity of \(f_{\lambda_i}(\tau)\). This asymmetry arises because positive and negative deviations in \(\tau\) alters the eigenvalues of \(\Sigma\) by different magnitudes. As a result, the corresponding elements of the covariance matrix \(\Sigma\) experience non-uniform magnitudes of change, with one side of the deviation spectrum potentially amplifying uncertainty more than the other.

Building on this observation, we now proceed to rigorously quantify the ambiguity set induced by uncertainty in the network parameters, focusing on how uncertainty in feedback gains translate into bounds on the covariance matrix.








\subsection{Ambiguity Due to Network Fluctuations}

Network fluctuations are common in practical networked systems, often arising from communication noise and varying environmental conditions. These fluctuations lead to uncertainty in the edge weights of the communication graph, which in turn induce ambiguity in the covariance matrix of the system observables. To capture this effect, we characterize an ambiguity set resulting from bounded fluctuations in the edge weights.

We consider two distinct scenarios separately: the case of zero time delay (\(\tau = 0\)) and the case of nonzero time delay (\(\tau \neq 0\)).

\subsubsection{Case 1: \(\tau = 0\)}

We model the uncertain edge weights \(\omega_{ij}\) with a nominal value \(\omega_{0,ij}\) and relative uncertainty bounded as
\begin{equation}\label{eqn:edge_weight_bound}
   (1 - \alpha_{\omega_{ij}}) \, \omega_{0,ij} \leq \omega_{ij} \leq (1 + \alpha_{\omega_{ij}}) \, \omega_{0,ij},
\end{equation}
where \(\alpha_{\omega_{ij}} \in [0,1)\) quantifies the level of uncertainty.

Exploiting the definition of the graph Laplacian in terms of unsigned incidence matrix,
\begin{equation}\label{eqn: Laplacian_using_incidence}
    L = RWR^T,
\end{equation}
we can directly relate changes in the edge weights to the ordering of Laplacian matrices. We state our next result in slightly more generalized form in terms of Laplacian matrix perturbation. 



\begin{proposition}[Covariance Bounds under Network Weight Perturbation for zero time-delay]\label{prop:ambiguity_edge_weight_tau_zero}
    Let $L_0$ be the nominal graph Laplacian with edge weights $\{\omega_{0,ij}\}$, and let $L$ be a perturbed Laplacian corresponding to edge weights $\{\omega_{ij}\}$ satisfying \eqref{eqn:edge_weight_bound}. 
Then, the perturbed Laplacian satisfies the multiplicative bound
\begin{equation}\label{eqn:Laplacian_multiplicative_bound}
   L_0 \left(1 - \|L_0^{\dagger}\Delta\|\right) \preceq L \preceq L_0 \left(1 + \|L_0^{\dagger}\Delta\|\right),
\end{equation}
where 
\[
\Delta = \sum_{i,j} \alpha_{\omega_{ij}} \omega_{0,ij} (e_i - e_j)(e_i - e_j)^\top
\]
quantifies the maximum deviation induced by edge weight uncertainty. Consequently, the ambiguity set of steady state distribution of observables is given by 
\begin{equation}\label{eqn:ambiguity_edge_weight}
    \mathcal{M}_{\omega} = \left\{ \mathbb{P} \sim \mathcal{N}(0, \Sigma) \;\middle|\; (1-\varepsilon_{\omega}^-) \Sigma_0 \preceq \Sigma \preceq (1+\varepsilon_{\omega}^+) \Sigma_0 \right\},
\end{equation}
where $\Sigma_0$ is the nominal covariance matrix, and
\begin{equation}
\varepsilon_{\omega}^- := \frac{\|\Delta L_0^{\dagger}\|}{1+\|\Delta L_0^{\dagger}\|},
\qquad 
\varepsilon_{\omega}^+ := \frac{\|\Delta L_0^{\dagger}\|}{1-\|\Delta L_0^{\dagger}\|},
\end{equation}
quantifies the induced uncertainty.
\end{proposition}

\subsubsection{Case 2: \(\tau \neq 0\)}
In the presence of nonzero time delay, the relationship between edge weight fluctuations and covariance matrix ordering becomes more complex due to the nonlinear dependence on the Laplacian eigenvalues and the time-delay function. While it is straightforward to establish uniform covariance bounds induced by uniform scaling of edge weights, the direct semidefinite ordering between covariance matrices corresponding to different weight sets generally does not hold.

We model the uncertain edge weight vector \(\bm{\omega} = [\omega_{ij}]_{(i,j) \in \mathcal{E}}\) as uniformly fluctuating around a nominal value \(\bm{\omega}_0\) with relative uncertainty bounded elementwise as
\begin{equation}\label{eqn:edge_weight_vector_bound}
   (1 - \alpha_{\bm{\omega}}) \, \bm{\omega}_{0} 
   \ \overset{\mathrm{e}}{\leq} \ 
   \bm{\omega} 
   \ \overset{\mathrm{e}}{\leq} \ 
   (1 + \alpha_{\bm{\omega}}) \, \bm{\omega}_{0},
\end{equation}
where the inequalities \(\overset{\mathrm{e}}{\leq}\) hold elementwise over all components of the weight vector \(\bm{\omega}\) and \(\alpha_{\bm{\omega}} = \alpha_{{\omega}} \bm{1}_n,\) where \(\alpha_{{\omega}} \in [0,1).\)

For notational convenience, define  
\begin{equation}\label{eqn:f_tau_lambda}
    f_{\tau}(\lambda_i) 
    = \frac{ \cos(\lambda_i \tau)}
           {\lambda_i  \left[1 - \sin(\lambda_i \tau)\right]},
\end{equation}
which represents the dependence of the function \(f\) on the \(i\)-th non-trivial Laplacian eigenvalue \(\lambda_i\) for a given time delay \(\tau\). A representative plot of \(f_{\tau}(\lambda_i)\) is shown in Fig.~\ref{fig:f_tau_lambda_graph}, which has a unique minimum. For the exposition of our next result, consider the following notation:
\begin{equation}\label{eqn:f_tau_i_plus_minus}
    \chi_{\tau}^{\pm} \left(\lambda_{0,i}\right) =   \frac{\left| 
    f_{\tau}\big((1 \pm \alpha_{\bm{\omega}})\lambda_{0,i}\big) - f_{\tau}(\lambda_{0,i}) 
  \right|}{f_{\tau}(\lambda_{0,i}) }
\end{equation}

 \begin{proposition}[Covariance Bounds under Network Weight Perturbation for non zero time-delay]\label{prop:ambiguity_edge_weight_tau_non_zero}
    Consider the uncertain edge weight vector \(\bm{\omega}\) satisfying the uniform bound in \eqref{eqn:edge_weight_vector_bound}.  
    The ambiguity set of steady-state covariance matrices is  
    \begin{equation}\label{eqn:ambiguity_edge_weight}
    \mathcal{M}_{\bm{\omega}} 
    = \left\{ \mathbb{P} \sim \mathcal{N}(0, \Sigma) 
        \ \middle| \ 
        (1-\varepsilon_{\bm{\omega}}^-) \, \Sigma_0 
        \preceq \Sigma 
        \preceq (1+\varepsilon_{\bm{\omega}}^+) \, \Sigma_0 
      \right\},
    \end{equation}
    where the bounds \(\varepsilon_{\bm{\omega}}^{\pm}\) depend on the spectral conditions as follows:
    \begin{align*}
        \varepsilon_{\bm{\omega}}^{\,\mp} 
&= \max_{i \in \{2,\dots,n\}} 
\chi_{\tau}^{\pm} \left(\lambda_{0,i}\right) \qquad \text{if} \quad \lambda_{0,n}^{+} \leq \bar{\lambda},\\
        \varepsilon_{\bm{\omega}}^{\,\mp} 
&= \max_{i \in \{2,\dots,n\}} 
\chi_{\tau}^{\pm} \left(\lambda_{0,i}\right) \qquad \text{if} \quad \lambda_{0,2}^{-} \geq \bar{\lambda},
    \end{align*}
and \(\lambda_{0,n}^+ = \lambda_{0,n} (1+\alpha_{\bm{\omega}}),\) \(\lambda_{0,2}^- = \lambda_{0,2} (1-\alpha_{\bm{\omega}}),\) and \(   \bar{\lambda} = \arg\min_{\lambda_i} f_{\tau}(\lambda_i)
    .\)
\end{proposition}

\begin{figure}[t]
    \centering
    \includegraphics[width=0.96\columnwidth]{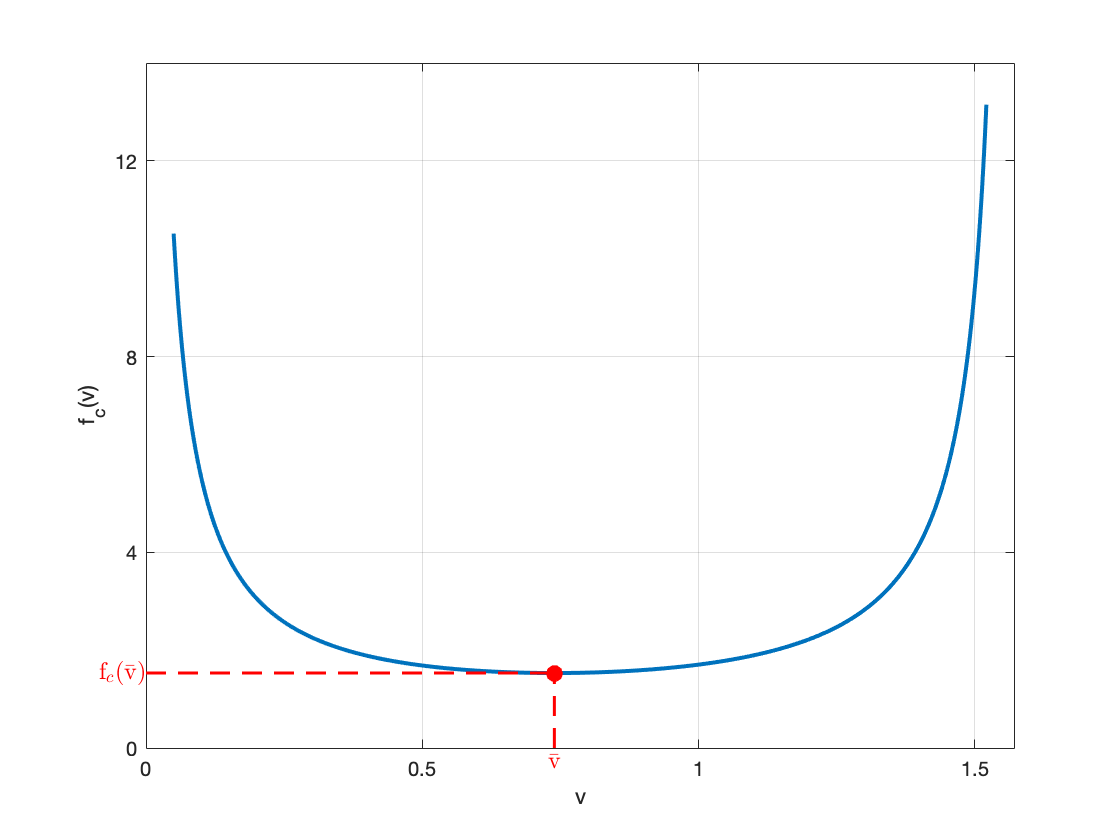} 
    \caption{Graph of \(f_c(v) = \frac{\cos(cv)}{2v (1- \sin(cv))}\).}
    \label{fig:f_tau_lambda_graph}
\end{figure}
A comprehensive analysis of ambiguity sets under arbitrary edge weight fluctuations with nonzero delay is challenging and beyond the scope of this work. We leave this direction for future investigation and focus here on representative cases with uniform weight perturbations.

\begin{remark}
All covariance matrices in the ambiguity set \(\mathcal{M}_{p}, \; p \in \{b, \tau\}\), are simultaneously diagonalizable, as they inherit the eigenvectors of the nominal covariance matrix. This property is essential for maintaining the analytical tractability of the distributionally robust formulation. 

In contrast, under arbitrary perturbations of the Laplacian weights with nonzero time delay, the covariance matrices corresponding to the nominal and perturbed Laplacians are no longer simultaneously diagonalizable. Without a shared eigenspace, direct semidefinite ordering between the two covariance matrices generally does not hold, making it difficult to establish uniform covariance bounds. The delay-free case ($\tau = 0$) is an exception: here, the covariance matrix is a linear function of the Laplacian pseudoinverse, which allows uniform bounds to be established, despite the loss of simultaneous diagonalizability.

For this reason, in the presence of nonzero time delay (\(\tau \neq 0\)), we restrict our ambiguity set to covariance matrices induced by uniform scaling of the Laplacian weights. This ensures that a semidefinite ordering can be established, enabling a tractable characterization of the distributionally robust constraints.

\end{remark}

\section{Distributionally Robust Risk of Cascading Failures}  \label{sec: dr_risk}

The quantification of risk under covariance uncertainty is well-established in domains such as portfolio optimization. In classical robust portfolio design, the covariance matrix of asset returns is often only partially known—typically bounded elementwise—and the objective is to compute the worst-case risk over all admissible covariance matrices that satisfy these bounds and the positive semidefiniteness constraint~\cite{boyd2004convex}. 

Motivated by this framework, we consider an analogous problem in multi-agent rendezvous, where uncertainty in system parameters (such as the diffusion coefficient, communication time delay, and edge weights) induces uncertainty in the covariance matrix of steady-state observables. Our goal is to quantify the worst-case, or distributionally robust (DR) risk, by explicitly accounting for ambiguity in the probability distribution of observables induced by model uncertainty. 

To this end, we first formalize the $\sigma$-algebra associated with the relevant events and present the conditional expectation formula that underpins our DR risk framework, followed by quantification of DR risk and a tractable optimization formulation for its computation.

\subsection{Conditional Expectation Based on Event Sigma-Algebra}

Consider a team of agents labeled \( \{1, \dots, n\} \). Our primary scenario analyzes the failure of the \(i\)-th agent to achieve \(c\)-consensus, defined by the event \( |y_i| > c \) as defined in \ref{def:systemic_failure}. Given this failure, we investigate the distributionally robust risk of failure for another agent \(j\).

To formalize the cascading failure event for the observable \(\bar{y}_i\), we define the set
\[
    U_{\delta^i} = (-\infty, -\delta^i - c) \cup (\delta^i + c, \infty),
\]
where \(c, \delta^i \in \mathbb{R}_+\) quantifies the magnitude of deviation of the observable \(y_i\). 

Let \(\mathcal{F}^i = \sigma(U_{\delta^i})\) denote the sigma-algebra generated by the set \(U_{\delta^i}\). Explicitly, \(\mathcal{F}^i = \{\emptyset, U_{\delta^i}, U_{\delta^i}^k, \mathbb{R}\}\), where \(U_{\delta^i}^k\) is the complement of \(U_{\delta^i}\). We also denote the trivial sigma-algebra as \(\mathcal{F}^{i,o} = \{\emptyset, \mathbb{R}\}\).

To quantify risk, we express the conditional expectation of the random variable \(|y_j|\) given that \(y_i\) lies in the failure set \(U_{\delta^i}\) in a familiar form:

\begin{lemma}\label{lem:conditional_expectation}
The conditional expectation of \(|y_j|\) given \(y_i \in U_{\delta^i}\) is given by
\[
    \mathbb{E}_{\mathbb{P}|\mathcal{F}^i}[|y_j|] = \frac{\mathbb{E}_{\mathbb{P}}[|y_j| \, \mathbf{1}_{\{y_i \in U_{\delta^i}\}}]}{\mathbb{P}[y_i \in U_{\delta^i}]},
\]
where \(\mathbf{1}_{\{y_i \in U_{\delta^i}\}}\) is the indicator function for the event \(y_i \in U_{\delta^i}\), which is 1 if \(y_i \in U_{\delta^i}\}\) and 0 otherwise.
\end{lemma}

This lemma serves as a critical step in quantifying the risk of cascading failures in multi-agent rendezvous, particularly by defining the conditional expectation of agent \( j \)’s deviation, given the information available from agent \( i \)’s observable. This formulation captures the inter dependencies between agents, enabling a deeper understanding of the cascading risk.

We refer to \cite{durrett2019probability} for details on \(\sigma\)-algebras and conditional expectation.

\subsection{Quantification of DR Risk}
Building on the distributionally robust optimization framework developed in \cite{Peyman2018, shapiro2022}, 
we now introduce a distributionally robust risk measure tailored to cascading failures in multi-agent systems. 
Specifically, we define the \emph{distributionally robust cascading risk} as follows:


\begin{equation}\label{eqn:condn_dr_risk_definition}
    \mathcal{R}_i \!\left[ \lvert y_j \rvert \right] 
    := \mathcal{R}^j_i 
    := \inf \left\{ \delta > 0 \;\middle|\; 
        \underset{\mathbb{P} \in {\mathcal{M}}_{p}}{\sup}~
        \mathbb{E}_{\mathbb{P} \mid \mathcal{F}^i} \!\left[ \lvert y_j \rvert \right] 
        \in U_{\delta^{+}}
    \right\},
\end{equation}
where \( U_{\delta^{+}} = \left( \delta + c , \infty \right) \).  
The measure \( \mathcal{R}^j_i \) in \eqref{eqn:condn_dr_risk_definition} 
represents the worst-case conditional risk of failure for agent \( j \), 
given that agent \( i \) has already failed, 
while accounting for distributional uncertainty characterized by the ambiguity set \(\mathcal{M}_p,\) for \(p \in \{b, \tau, \omega\}.\)

\begin{remark}
Note that the risk measure \(\mathcal{R}^j_i\) in \eqref{eqn:condn_dr_risk_definition} is a convex risk measure. 
This property follows directly from the triangle inequality and the linearity of the expectation operator \cite{durrett2019probability}.
\end{remark}

 We now provide a closed-form characterization of the distributionally robust cascading risk.

For the exposition of the next result, we introduce the following notations:
\[
\mathbb{E}^j_i = \mathbb{E}_{\mathbb{P}|\mathcal{F}^i}\left[\lvert y_j \rvert \right], \quad 
h(\cdot) = \text{erf}(\cdot), \quad 
\bar{\delta} = \delta^i + c.
\]



\begin{theorem}\label{thm:DR_risk}
Suppose the system \eqref{eqn: network_dynamics} has reached a steady state, and agent \(i\) has failed to achieve \(c\)-consensus, with its observable \(y_i\) belonging to the uncertainty set \(U_{\delta^i}\). The distributionally robust cascading risk of agent \(j\) is given by  
\begin{equation}\label{eqn:DR_risk}
    \mathcal{R}^j_i =
    \begin{cases}
       0, & \text{if} ~ \underset{\mathbb{P} \in {\mathcal{M}}_{p}}{\sup}~ \mathbb{E}^j_i \leq c, \\
        \underset{\mathbb{P} \in {\mathcal{M}}_{p}}{\sup}~ \mathbb{E}^j_i - c, & \text{otherwise},
    \end{cases}
\end{equation}
where \(\mathbb{E}^j_i\), the expected value of \(y_j\) given that \(y_i \in U_{\delta^i}\), is given by  
\begin{equation}\label{eqn:conditional_expectation_closed_form}  
        \mathbb{E}^j_i =  \sqrt{\frac{2}{\pi}} \frac{\sigj}{1-h(\delta^{*})} \left[1 - h\left(\frac{\delta^{*}}{\rho'}\right) +\rho h\left(\frac{\rho \delta^{*}}{\rho'}\right)\textnormal{e}^{-{\delta^{{*}^2}}}\right],
\end{equation}  
where \({\delta^{*}} = \frac{\bar{\delta}}{\sqrt{2}\sigi}\), and \(\rho, \rho'\) are as defined in Lemma \ref{lem:bivariate_normal}.  
\end{theorem}
The two cases in \eqref{eqn:DR_risk} are straightforward: \(\mathcal{R}^j_i = 0\) indicates that the expected observable of agent \(j\) remains within the tolerance \(c\), while a nonzero risk reflects failure to achieve \(c\)-consensus in expectation. From \eqref{eqn:DR_risk}, the risk depends on the time delay, graph topology, and agent \(i\)'s failure mode, quantified by \(\bar{\delta}\). A higher expected value leads to greater risk, enabling network design to minimize the risk of cascading failures.

Taking the supremum over all distributions in the ambiguity set ensures a worst-case assessment, providing a robust measure of vulnerability to cascading failures.

An immediate consequence of Theorem \ref{thm:DR_risk} is the case of independent agents.  

\begin{corollary}\label{cor:single_risk_rho_0}
If agents \(j\) and \(i\) are uncorrelated, the distributionally robust cascading risk of agent \(j\) reduces to the single-agent failure risk:
\begin{equation}\label{eqn:DR_risk_single}
    \mathcal{R}^j  =
    \begin{cases}
    0, & \text{if} ~  \sigma_{0,j} \leq \sqrt{\frac{\pi}{2}}\frac{c}{1+\varepsilon_p} \\
     \sqrt{\frac{2}{\pi}} \sigma_{0,j}(1+ \varepsilon_p) - c,  & \text{otherwise} 
    \end{cases}
\end{equation}
\end{corollary}


This result quantifies the inherent failure risk of a single agent, corresponding to the case where agent \(i\)'s information is captured by the trivial \(\sigma\)-algebra \(\mathcal{F}^{i,o}\).

The distributionally robust cascading risk vector of all agents can be written as
\begin{equation}\label{eqn:risk_vector}
    \bm{\mathcal{R}}_i = \left[ \mathcal{R}^1_i, \dots,  \mathcal{R}^n_i \right],
\end{equation}
where \(\mathcal{R}^j_i =   \mathcal{R}^i \)  for \(i = j\).

Computing the DR cascading risk in Theorem \ref{thm:DR_risk} directly via a supremum over the ambiguity set can be challenging. To address this, we reformulate the problem as a tractable optimization over individual agents’ variances and pairwise correlations. 

\subsection{Optimization Formulation of DR Risk}
The DR risk expression in Theorem~\ref{thm:DR_risk} is posed as an optimization problem over the covariance matrix, constrained to lie within the positive semidefinite cone. 
We reformulate this problem into an equivalent representation in which the DR risk is computed via a constrained optimization over scalar variables—namely, the individual variances and pairwise correlations. 
This reformulation is formalized in Proposition~\ref{prop:opt_formulation_over_ambiguity_set}.

\begin{proposition}[Optimization over Ambiguity Set]\label{prop:opt_formulation_over_ambiguity_set}
The distributionally robust cascading risk stated in Theorem \ref{thm:DR_risk} can be formulated as the following constrained maximization problem:
\begin{align}
\underset{\sigma_i, \sigma_j, \rho}{\text{maximize}} \quad 
& \Ei \left[ \left |y_j\right|\right] - c  \label{eqn:opt_objective}\\
\text{subject to} \quad
& \sqrt{(1 - \varepsilon_p^-)} \, \sigma_{i,0} \leq \sigma_i \leq \sqrt{(1 + \varepsilon_p^+)} \, \sigma_{i,0},  \label{eqn:opt_constraint_sigi}\\
& \sqrt{(1 - \varepsilon_p^-)} \, \sigma_{j,0} \leq \sigma_j \leq \sqrt{(1 + \varepsilon_p^+)} \, \sigma_{j,0},  \label{eqn:opt_constraint_sigj}\\
& 0 \,  \leq \mid\rho\mid < 1 \label{eqn:opt_constraint_rho}. 
\end{align}
\end{proposition}
Proposition~\ref{prop:opt_formulation_over_ambiguity_set} reformulates the DR risk problem over the ambiguity set into an optimization problem in the variables \(\sigma_i\), \(\sigma_j\), and \(\rho\), subject to bounded constraints arising from the boundedness of the covariance matrix within the positive semidefinite cone. This yields a tractable representation of the DR risk, where the objective is to determine the maximum value of the objective function over these feasible bounds. This result further underscores the significance of our formulation for quantifying the ambiguity set, wherein covariance matrices lie within a cone and admit simultaneous diagonalizability. These two properties directly enable us to express the constraints in the optimization problem. 
Most importantly, the fact that the correlation coefficient \(\rho\) may vary across the entire range \([0,1)\) highlights how parameter uncertainty can significantly alter the correlation between agents \(i\) and \(j\). This variability further motivates the need for a distributionally robust risk framework that accounts for such uncertainty.


To address numerical instability arising from the denominator term \( \textnormal{erfc}(\delta^*) \) approaching zero rapidly in \eqref{eqn:DR_risk}, we adopt an analytical approximation. Several approximations and bounds for the complementary error function have been proposed in the literature~\cite{karagiannidis2007improved, tanash2021improved_error_function_coeff}. We employ the following approximation \cite{karagiannidis2007improved} for \(x>0\), chosen for its simplicity.
\begin{equation}\label{eqn:erfc_approximation}
    \textnormal{erfc}(x) \approx \frac{(1 - e^{-Ax}) e^{-x^2}}{B \sqrt{\pi}x},
\end{equation}
where the constants are chosen as \( A = 1.98 \) and \( B = 1.135 \). Substituting this approximation into the expression in~\eqref{eqn:conditional_expectation_closed_form}, we obtain:

\begin{equation}\label{eqn:expectation_approx}
    \mathbb{E}_i^j 
    = \sqrt{\tfrac{2}{\pi}}\, \sigma_j 
      \Big( H_1(\delta^*) + H_2(\delta^*) \Big),
\end{equation}
where
\begin{align}
    H_1(\delta^*) 
    &= \rho' \, 
       \frac{1 - \exp\!\left(-A \tfrac{\delta^*}{\rho'}\right)}
            {1 - \exp\!\left(-A \delta^*\right)} \,
       \exp\!\left(-\tfrac{\rho^2}{\rho'^2} \, \delta^{*2}\right), 
       \label{eqn:H_1_delta_star} \\[6pt]
    H_2(\delta^*) 
    &= \frac{\rho B \sqrt{\pi} \, \delta^*}
            {1 - \exp\!\left(-A \delta^*\right)} \,
       h\!\left(\tfrac{\rho \delta^*}{\rho'}\right).
       \label{eqn:H_2_delta_star}
\end{align}
The approximation in \eqref{eqn:erfc_approximation} is valid in our setting, since \(\delta^* > 0\) in \eqref{eqn:conditional_expectation_closed_form}.
While the effect of \(|\rho|\) can also be analyzed directly from 
\eqref{eqn:conditional_expectation_closed_form}, we emphasize it here through 
the approximate form in \eqref{eqn:expectation_approx}. Specifically, 
\(H_1(\delta^*)\) is a decreasing function of \(|\rho|\), where
\(H_2(\delta^*)\) is an increasing function of \(|\rho|\). Consequently, as 
\(|\rho|\) grows, the relative importance of \(H_1(\delta^*)\) and 
\(H_2(\delta^*)\) shifts, further underscoring the role of the constraint on 
\(|\rho|\) in \eqref{eqn:opt_constraint_rho}.

\begin{remark}\label{rem:opt_constraints_b}
In the special case where \(\mathcal{M}_p = \mathcal{M}_b\), the constraints \eqref{eqn:opt_constraint_sigi}, \eqref{eqn:opt_constraint_sigj}, and \eqref{eqn:opt_constraint_rho} simplify to
\[
\sigma_{\star} = \sqrt{(1 + \theta)}\, \sigma_{\star,0} \quad \text{for } \star \in \{i, j\}, \quad \theta \in [-\alpha_b, \alpha_b],
\]
and \(\rho = \rho_0\). This simplification follows directly from the definition of the ambiguity set in \eqref{eqn:ambiguity_set_b}.
\end{remark}


\section{Fundamental Limits}\label{sec:fundamental_limits}

Our primary goal in this work is to inform the design of networks that are resilient to cascading large fluctuations under worst-case scenarios. To this end, we have quantified the distributionally robust (DR) risk of such fluctuations in Theorem~\ref{thm:DR_risk}. However, the DR risk expression involves a complex interplay of each agent’s variance and their pairwise correlations, making direct analysis challenging. In this section, we derive upper and lower bounds on the DR risk, which provide actionable insights into network vulnerability and simplify the design of resilient networks. Since these bounds are expressed in terms of network eigenvalues and system parameters such as time delays, they apply to any pair of agents $i$ and $j$. In this sense, the bounds are global: rather than computing the DR risk for each agent pair individually, they offer uniform guarantees across the entire network, enabling efficient risk assessment and guiding principled network design.

For the exposition of our upcoming results, we define several auxiliary quantities related to the eigenstructure of the covariance matrix \( \Sigma_0 \). Specifically, we denote the largest and smallest non-zero eigenvalues of the covariance matrix \( \Sigma_0 \), evaluated at the nominal parameters \( p_0 = \{b_0, \tau_0, \omega_0\} \), as:
\begin{align}\label{eqn:psi_extremal}
\psi_{0,n} &:= \max_{i \in \{2, \dots, n\}} b^2\, \tau f(\lambda_i\tau)|_{p = p_0}, \\
\psi_{0,2} &:= \min_{i \in \{2, \dots, n\}} b^2\, \tau f(\lambda_i\tau)|_{p = p_0}.
\end{align}


We define modified versions of the extremal eigenvalues that incorporate the size of the network. Specifically, we introduce:
\begin{equation}\label{eqn:psi_tilde}
\tilde{\psi}_{0,n} := \psi_{0,n} \left(1 - \frac{1}{n} \right), \qquad
\tilde{\psi}_{0,2} := \psi_{0,2} \left(1 - \frac{1}{n} \right).
\end{equation}

To incorporate parameter induced bounds, we introduce perturbed versions of relevant quantities. The bounds on standard deviations and extremal eigenvalues are given by:
\begin{equation}\label{eqn: sigma_+_-}
\sigma_{0,\star}^{+} := \sigma_{0,\star} \sqrt{1 + \varepsilon_p^{+}}, \qquad \star \in \{i, j\},
\end{equation}
\begin{equation}\label{eqn: psi_+_-}
\tilde{\psi}_{0,\star}^{+} := \tilde{\psi}_{0,\star} (1 + \varepsilon_p^{+}), \qquad \star \in \{2, n\},
\end{equation}
where \( \varepsilon_p^{\pm} \) represent upper and lower bounds on the parameter perturbations. These definitions will be used repeatedly in the derivation of risk bounds and performance guarantees.

\subsection{General Parameter Bounds}
We begin by establishing general bounds on the distributionally robust (DR) risk as a function of the eigenvalues of the covariance matrix associated with the uncertainty in three key system parameters: diffusion rate, time-delay, and network-related parameters. These bounds depend explicitly on the bounds of the ambiguity set \(\varepsilon^{+}_p\) , which quantifies the level of uncertainty. Following this, we refine the analysis to isolate the contribution of the network-induced uncertainty, highlighting how the spectral properties of the underlying graph influence the DR risk.
To state our next result, we introduce the following notation:
\begin{equation}\label{eqn:kappa_x}
    \kappa(x) = \frac{B\,\bar{\delta}_i}{1 - \exp\!\left(-A\,\frac{\bar{\delta}_i}{\sqrt{2}\,x}\right)}\,,
\end{equation}
where \(A\) and \(B\) are as defined in \eqref{eqn:erfc_approximation}.
 \begin{theorem}\label{thm:dr_risk_upper_bound}
 For all \(i,j \in \{1,\dots,n\}\) with \(i \neq j\), the distributionally robust cascading risk
\eqref{eqn:condn_dr_risk_definition} of agent \(j\) conditioned on the observable of agent \(i\),
\(y_i \in U_{\delta_i}\), satisfies the following upper bound:
     \begin{equation}\label{eqn:dr_risk_upper_bound}
                                \underset{\underset{i,j \in \{1, \dots, n\}}{i \neq j}}{\sup}\mathcal{R}^j_i \leq \sqrt{\tilde{\psi}_{0,n}^+}~\mathcal{R}^u - c,
     \end{equation}
     \noindent where 
     \begin{align*}
\mathcal{R}^u =
\begin{cases}
|\rho_0|\,
\kappa\!\bigl(\sqrt{\tilde{\psi}_{0,n}^+}\bigr)
+ \sqrt{2/\pi}\,\rho_0',
& \mathcal{M}_p = \mathcal{M}_b,
\\[0.6em]
\sqrt{
    \kappa^2\!\bigl(\sqrt{\tilde{\psi}_{0,2}^+}\bigr)
    + 2/\pi
},
& \mathcal{M}_p \in \{\mathcal{M}_\tau,\,\mathcal{M}_\omega\},
\end{cases}
\end{align*}
         and \(\Tilde{\psi}_{0,n}^+ \) is as defined in \eqref{eqn: psi_+_-}, and \(\kappa(\cdot)\) as defined in \eqref{eqn:kappa_x}.
 \end{theorem}

Theorem~\ref{thm:dr_risk_upper_bound} characterizes how the upper bound of the distributionally robust cascading risk depends on the maximum and minimum non-zero eigenvalues of the covariance matrix as well as the failure mode of agent \(i\). More importantly, it quantifies how parameter uncertainty exacerbates this upper bound, through the radius of the ambiguity set, 
captured by \(\tilde{\psi}_{0,n}^+\).

While the upper bound of the distributionally robust (DR) risk provides insights into the maximum achievable risk, a more important and informative result is the quantification of the lower bound on DR risk. This result allows us to establish a fundamental limit on the minimum achievable cascading risk under bounded parameter uncertainty. We formalize this result in the following theorem.

 \begin{theorem}\label{thm:dr_risk_lower_bound}
  For all \(i,j \in \{1,\dots,n\}\) with \(i \neq j\), the distributionally robust cascading risk
\eqref{eqn:condn_dr_risk_definition} of agent \(j\) conditioned on the observable of agent \(i\),
\(y_i \in U_{\delta_i}\), satisfies the following lower bound:
     \begin{equation}\label{eqn:dr_risk_lower_bound}
     \underset{\underset{i,j \in \{1, \dots, n\}}{i \neq j}}{\sup}\mathcal{R}^j_i 
\;\geq\; 
\sqrt{\tilde{\psi}_{0,2}^+}\,
\mathcal{R}^l - c,
     \end{equation}
     \noindent where 
               \begin{align*}
\mathcal{R}^l =
\begin{cases}
\lvert \rho_0 \rvert\,\kappa\!\left({\sqrt{\tilde{\psi}_{0,2}^+}}\right) 
        - \sqrt{\tfrac{2}{\pi}}\, \rho'_0,
& \mathcal{M}_p = \mathcal{M}_b,
\\[0.6em]
\max\!\left\{ \sqrt{\frac{2}{\pi}}, \;\kappa\!\left({\sqrt{\tilde{\psi}_{0,2}^+}}\right) - \eta \right\},
& \mathcal{M}_p \in \{\mathcal{M}_\tau,\,\mathcal{M}_\omega\},
\end{cases}
\end{align*}
     \noindent  \(\Tilde{\psi}_{0,2}^+\) is as defined in \eqref{eqn: psi_+_-}, \(\kappa(\cdot)\) as defined in \eqref{eqn:kappa_x} and \(\eta = 2\sqrt{\frac{\theta}{\pi}} \rightarrow 0\) as \(\theta \rightarrow 0\) as \(|\rho| \rightarrow 1\).
 \end{theorem}

An immediate corollary of Theorems~\ref{thm:dr_risk_upper_bound} and \ref{thm:dr_risk_lower_bound} is the following result, which characterizes the distributionally robust risk of failure for a single agent:
\begin{corollary}\label{cor:single_agent_upper_lower_bound}
    For the case of independent agents, the distributionally robust risk satisfies the following bounds
    \begin{equation}\label{eqn:single_agen_upper_bound}
                   \sqrt{\tfrac{2}{\pi} ~ 
\tilde{\psi}_{0,2}^+}  - c  \leq      \underset{j \in \{1, \dots, n\}}{\sup}\mathcal{R}^j \leq \sqrt{\tfrac{2}{\pi} ~ 
\tilde{\psi}_{0,n}^+}  - c .
    \end{equation}
\end{corollary}
Corollary~\ref{cor:single_agent_upper_lower_bound} quantifies the bounds on the DR risk of failure for a single agent in terms of the maximum and minimum nontrivial eigenvalues of the covariance matrix. This relationship is crucial for establishing bounds in terms of network parameters, as illustrated next.

    

\subsection{Network Induced Fundamental Limit}
Having established bounds on the DR risk in terms of the eigenvalues of the covariance matrix—reflecting general uncertainty in system observables—we now derive explicit bounds that isolate the role of network parameters and network-induced uncertainty. In this part, we consider uncertainty solely in the edge weights of the network. While the time-delay and diffusion coefficient are intrinsic properties of the system dynamics, the network structure serves as the primary design lever, playing a central role in the synthesis of control laws for risk-resilient networked systems.

While it is possible to derive both upper and lower bounds induced by network uncertainty, our primary interest lies in understanding the fundamental performance limits imposed by the network structure. To this end, we focus exclusively on the minimum achievable DR risk, which we formalize in the following result. For better readability and brevity, we introduce the following notation for \(x > 0\):
\begin{equation}\label{eqn:R_l_f}
  \mathcal{R}^l_f(x) =   \bar{f}^+\!\left(x\right)\,
\max\!\left\{ \sqrt{\tfrac{2}{\pi}},\;
\bar{\kappa}\!\left(\bar{f}^+\!\left(x\right)\right)\right\} - c,
\end{equation}
where \(\bar{f}^+\!\left(x\right) = b_0 \sqrt{ f_{\tau_0}\left(x\right) (1+\varepsilon_{\omega}^+)},\) and \(f_{\tau_0}(\cdot)\) as defined in \eqref{eqn:f_tau_lambda}, \(\bar{\kappa}(\cdot) =\kappa(\cdot) - \eta,\) where \(\kappa(\cdot)\) is as defined in \eqref{eqn:kappa_x}.

\begin{theorem}\label{thm:fundamental_limit}
    The DR cascading risk as function of network parameter and network uncertainty satisfies
            \begin{equation}
\underset{\underset{i,j \in \{1, \dots, n\}}{i \neq j}}{\sup}\, \mathcal{R}^j_i 
\;\geq\;
\begin{cases}
\begin{aligned}
&  \mathcal{R}^l_f(\lambda_{0,n}) 
&\quad\text{if } \lambda_{0,n}^+ \leq \bar{\lambda},
\end{aligned} \\[1.2em]
\begin{aligned}
& \mathcal{R}^l_f(\lambda_{0,2}) 
&\quad\text{if } \lambda_{0,2}^- > \bar{\lambda}.
\end{aligned}
\end{cases}
\end{equation}
    \noindent where \(\lambda_{0,n}^+,\) \(\lambda_{0,2}^-\) and \(\bar{\lambda}\) as defined in Proposition \ref{prop:ambiguity_edge_weight_tau_non_zero}.
\end{theorem}

Theorem~\ref{thm:fundamental_limit} effectively implies that the minimum achievable DR cascading risk is constrained by the maximum and minimum nontrivial eigenvalues (the Fiedler eigenvalue \cite{Fiedler1973}) of the graph Laplacian. Specifically, when \(\lambda_{0,n}^+ \leq \bar{\lambda}\), the minimum achievable risk decreases as the largest eigenvalue of the Laplacian increases, whereas when \(\lambda_{0,2}^- \geq \bar{\lambda}\), the minimum achievable DR cascading risk decreases as the Fiedler eigenvalue becomes smaller. This is a direct consequence of the monotonicity of the function \(f_{\tau}(\lambda)\), as defined in~\eqref{eqn:f_tau_lambda} and illustrated in Fig.~\ref{fig:f_tau_lambda_graph}. 
For the special case of zero time delay, \(\tau = 0\), we obtain
\begin{equation}\label{eqn:R_l_f_tau_0}
  \underset{\substack{i,j \in \{1,\dots,n\}\\ i \neq j}}{\sup}\, \mathcal{R}^j_i
  \;\geq\;
  \sqrt{\frac{1+\varepsilon_{\omega}}{\lambda_{0,n}}}
  \max\!\left\{
      \sqrt{\tfrac{2}{\pi}},\;
      \bar{\kappa}\!\left(  \sqrt{\frac{1+\varepsilon_{\omega}}{\lambda_{0,n}}}\right)
  \right\}
  - c,
\end{equation}
which follows from Remark~\ref{rem: Sigma_y_time_delay_zero}, and shows a clear dependence on the largest eigenvalue and the uncertainty radius of the edge weights.


An immediate corollary for the single-agent DR risk of failure follows along the same lines as Corollary~\ref{cor:single_agent_upper_lower_bound}, expressed in terms of the network parameters and the network uncertainty.

\begin{corollary} \label{cor:fundamental_limit_single_agent}
The distributionally robust (DR) risk of failure of a single agent, viewed as a function of the network parameters and the network uncertainty, satisfies
\begin{equation}
\underset{j \in \{1,\dots,n\}}{\sup}\, \mathcal{R}^j 
\;\geq\;
\sqrt{\frac{2}{\pi}}
\begin{cases}
\begin{aligned}
& \bar{f}^+(\lambda_{0,n}) - c,
& \text{if } \lambda_{0,n}^+ \leq \bar{\lambda},
\end{aligned}
\\[1.2em]
\begin{aligned}
& \bar{f}^+(\lambda_{0,2}) - c,
& \text{if } \lambda_{0,2}^- > \bar{\lambda},
\end{aligned}
\end{cases}
\end{equation}
where $\lambda_{0,n}^+$, $\lambda_{0,2}^-$, and $\bar{\lambda}$ are as defined in Proposition~\ref{prop:ambiguity_edge_weight_tau_non_zero}.
\end{corollary}

The dependence of the single-agent DR risk of failure on the network parameters is qualitatively similar to the cascading case discussed earlier. An instructive special case arises when the time-delay parameter is set to zero. In this case,
\begin{equation}\label{eqn:R_j_tau_0}
  \underset{j \in \{1,\dots,n\}}{\sup}\, \mathcal{R}^j
  \;\geq\;
  \sqrt{\frac{2}{\pi}\,\frac{1+\varepsilon_{\omega}}{\lambda_{0,n}}}\; - c.
\end{equation}

In particular, setting $c = 0$ yields the relation
\begin{equation}\label{eqn:R_j_tau_0_hyperbolic_form}
  \underset{j \in \{1,\dots,n\}}{\sup}\, \mathcal{R}^j\, \sqrt{\lambda_{0,n}}
  \;\geq\;
  \sqrt{\frac{2}{\pi}(1+\varepsilon_{\omega})}.
\end{equation}

Expression~\eqref{eqn:R_j_tau_0_hyperbolic_form} clearly shows that the product of the DR risk and the square root of the largest eigenvalue of the graph Laplacian is bounded below by a constant determined by the ambiguity radius induced by the edge-weight uncertainty of the underlying network.

\begin{figure*}[t]
\centering

\makebox[0.23\textwidth]{\textbf{Complete}}%
\makebox[0.23\textwidth]{\textbf{14-cycle}}%
\makebox[0.23\textwidth]{\textbf{6-cycle}}%
\makebox[0.23\textwidth]{\textbf{Path}}\\[0.5em]

\raisebox{0.75\height}{\rotatebox{90}{Diff. Coeff.}}%
\begin{subfigure}[b]{0.23\textwidth}
    \includegraphics[width=\linewidth]{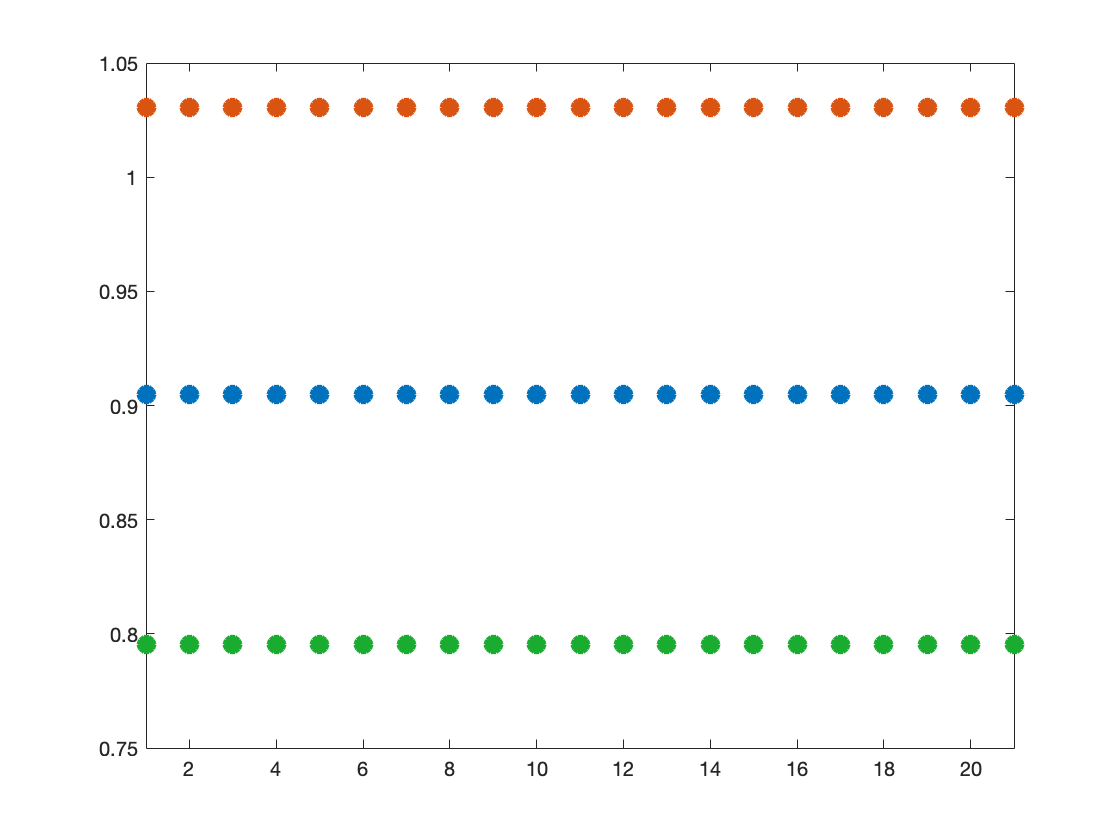}
    \subcaption{}
\end{subfigure}%
\begin{subfigure}[b]{0.23\textwidth}
    \includegraphics[width=\linewidth]{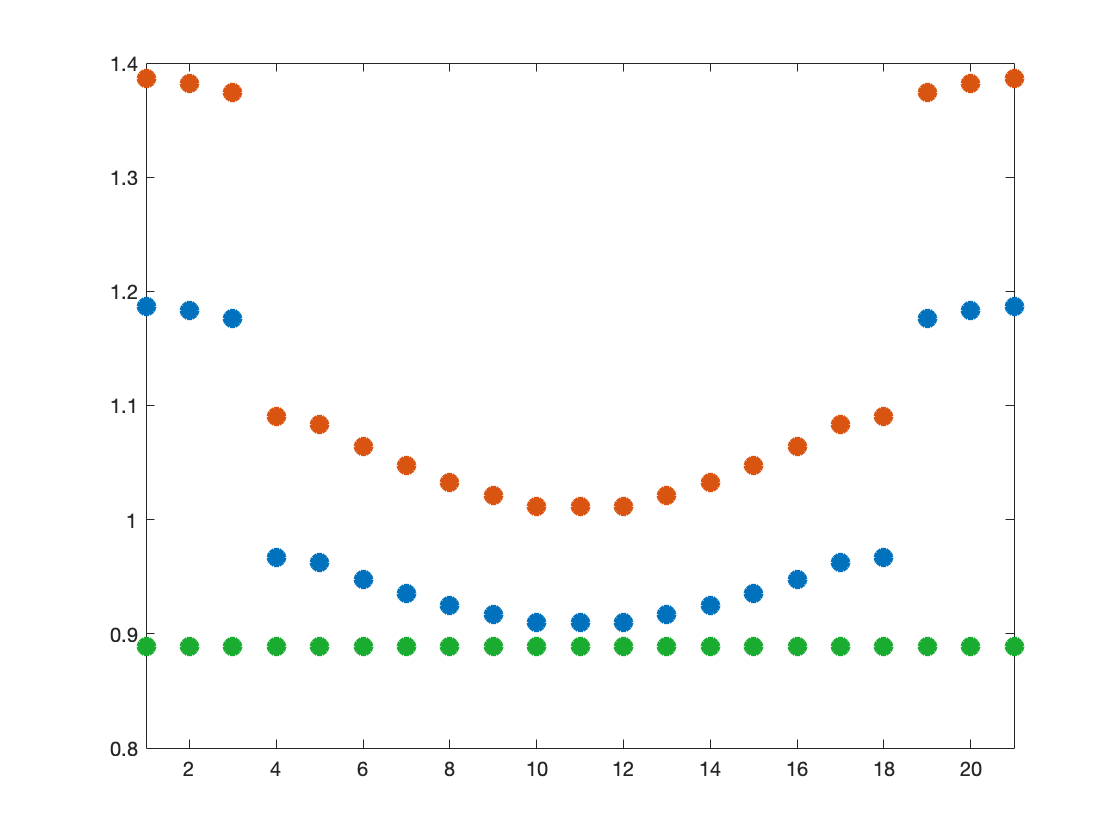}
    \subcaption{}
\end{subfigure}%
\begin{subfigure}[b]{0.23\textwidth}
    \includegraphics[width=\linewidth]{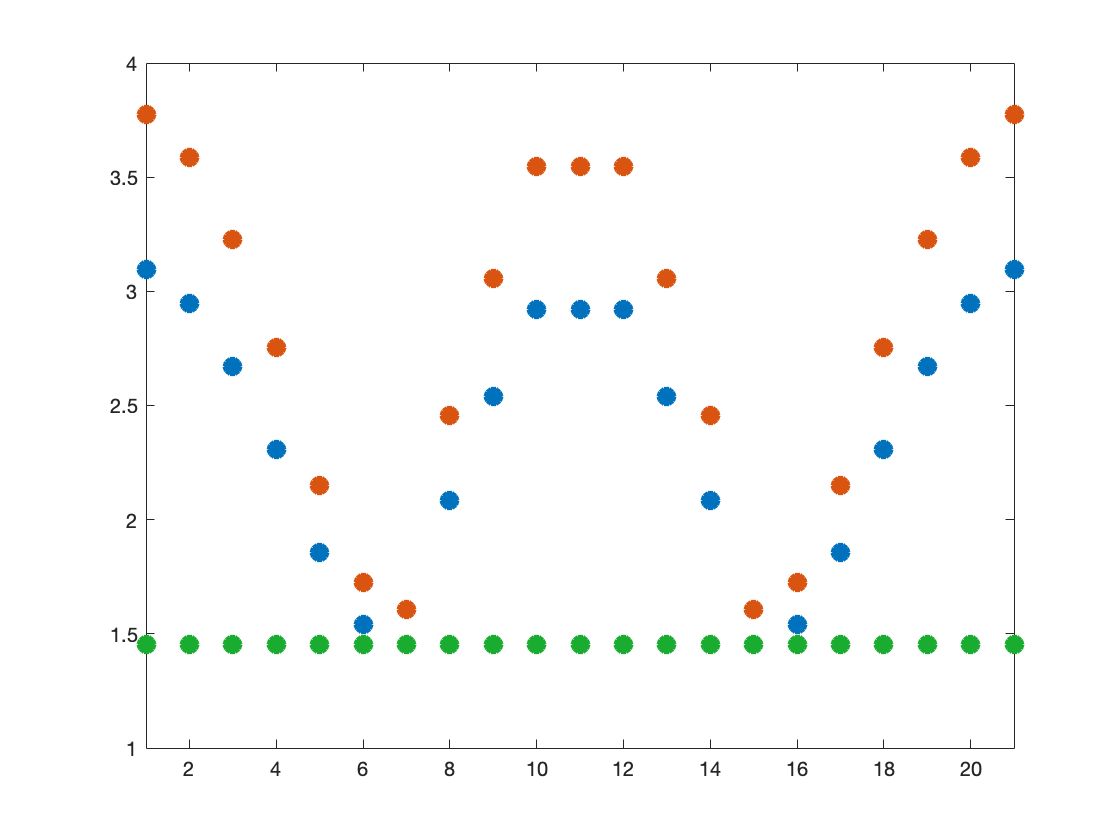}
    \subcaption{}
\end{subfigure}%
\begin{subfigure}[b]{0.23\textwidth}
    \includegraphics[width=\linewidth]{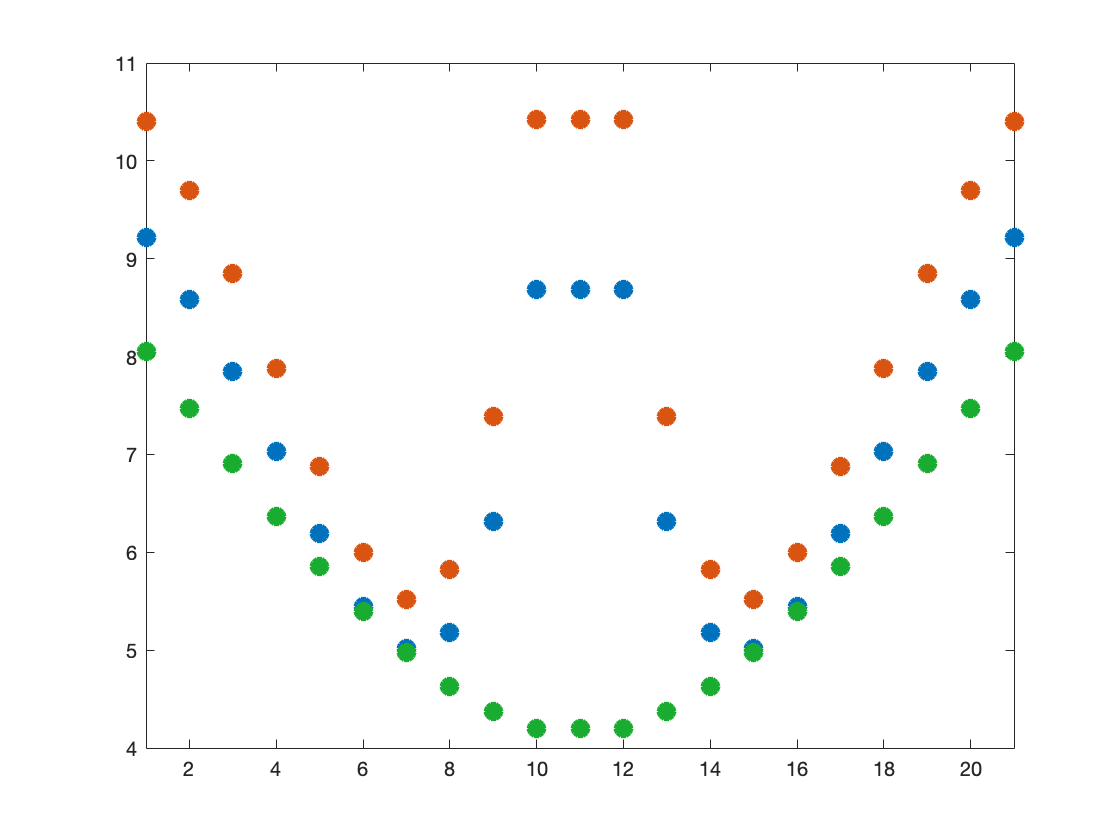}
    \subcaption{}
\end{subfigure}\\[0.5em]

\raisebox{0.7\height}{\rotatebox{90}{Time-Delay}}%
\begin{subfigure}[b]{0.23\textwidth}
    \includegraphics[width=\linewidth]{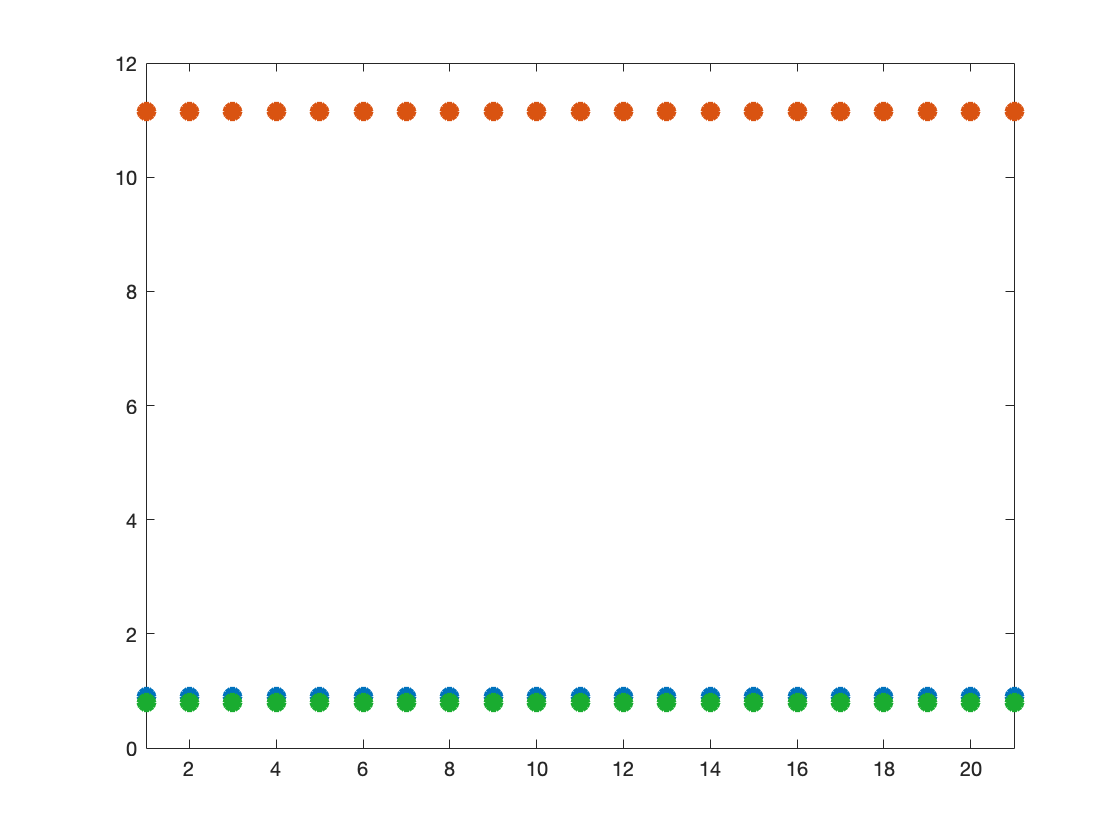}
    \subcaption{}
\end{subfigure}%
\begin{subfigure}[b]{0.23\textwidth}
    \includegraphics[width=\linewidth]{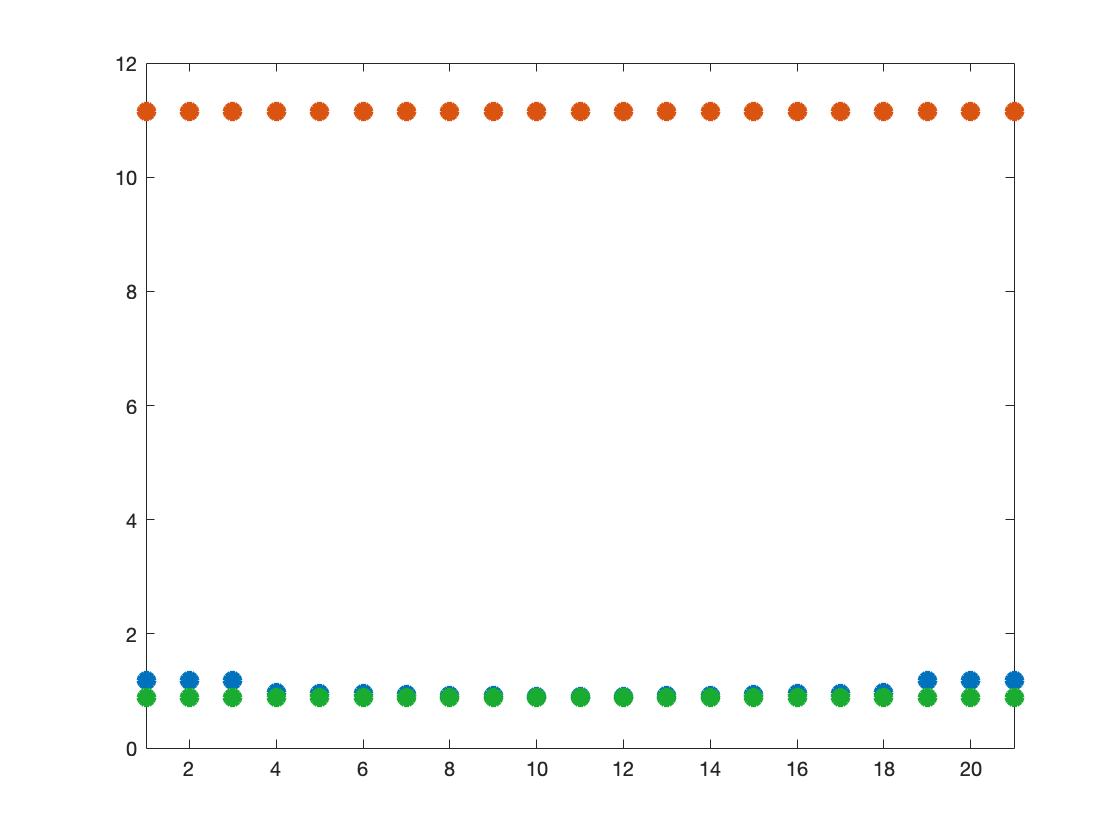}
    \subcaption{}
\end{subfigure}%
\begin{subfigure}[b]{0.23\textwidth}
    \includegraphics[width=\linewidth]{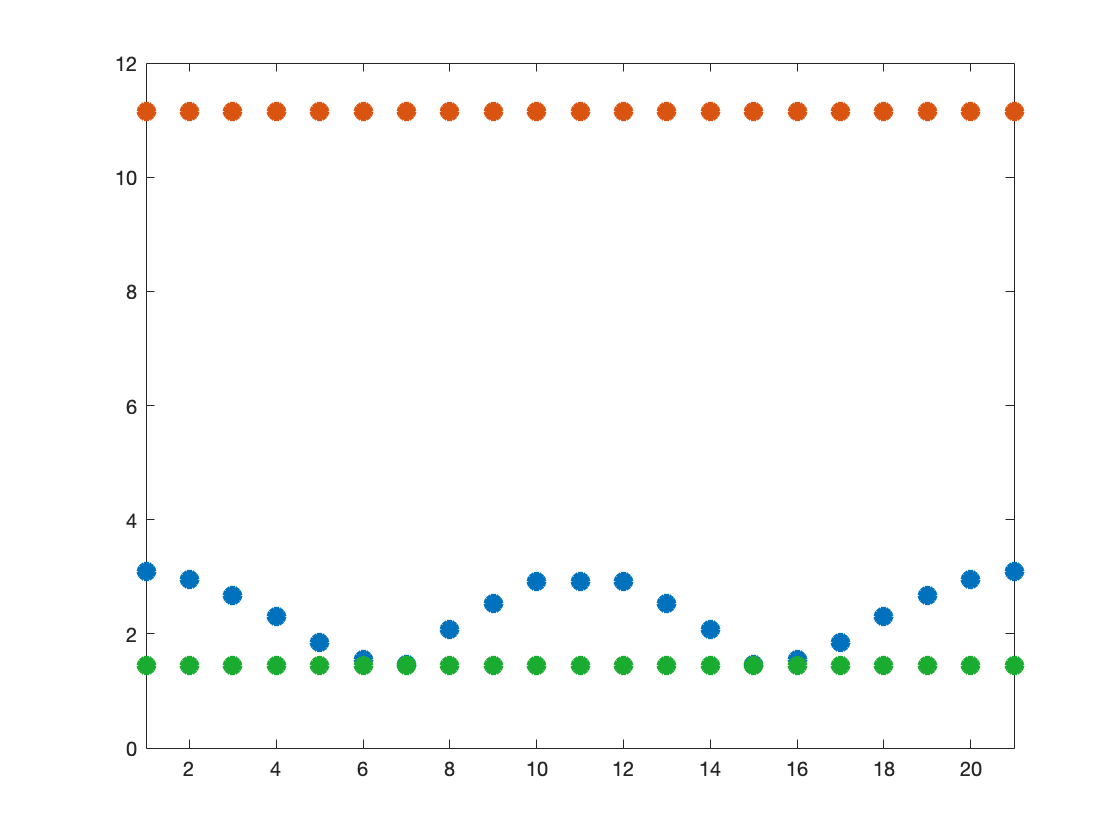}
    \subcaption{}
\end{subfigure}%
\begin{subfigure}[b]{0.23\textwidth}
    \includegraphics[width=\linewidth]{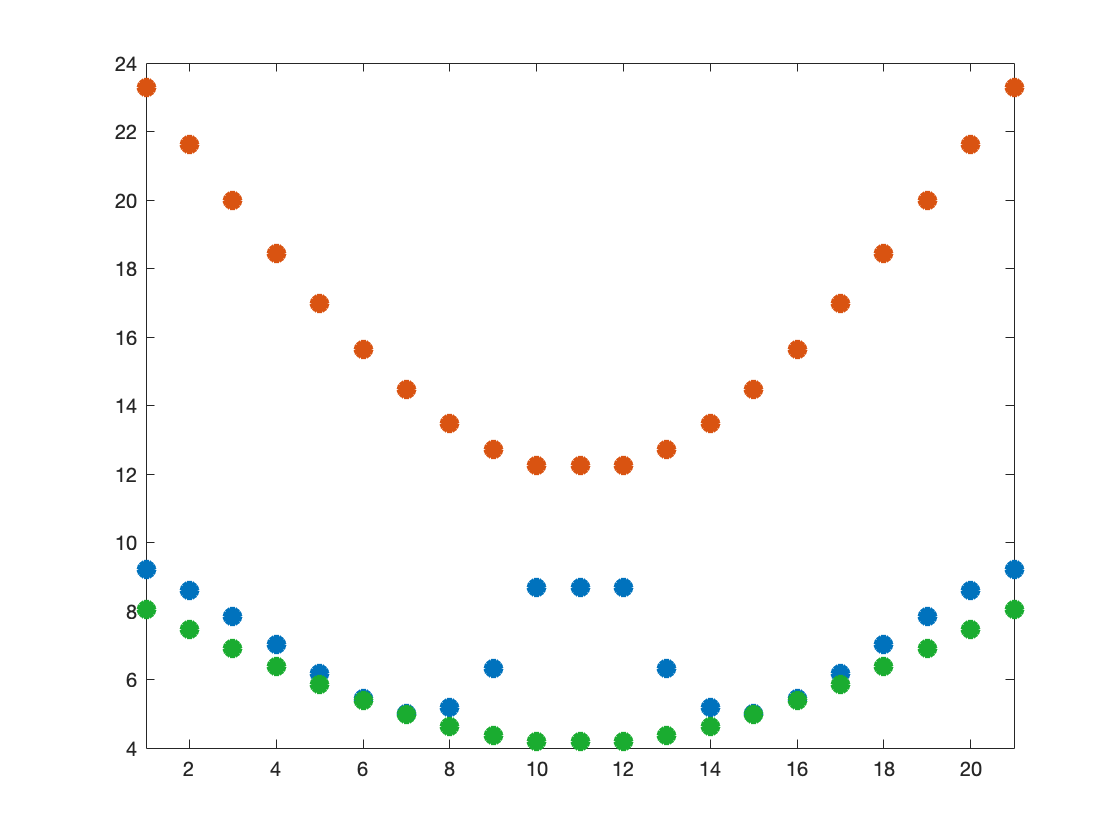}
    \subcaption{}
\end{subfigure}\\[0.5em]

\raisebox{0.6\height}{\rotatebox{90}{Edge Weight}}%
\begin{subfigure}[b]{0.23\textwidth}
    \includegraphics[width=\linewidth]{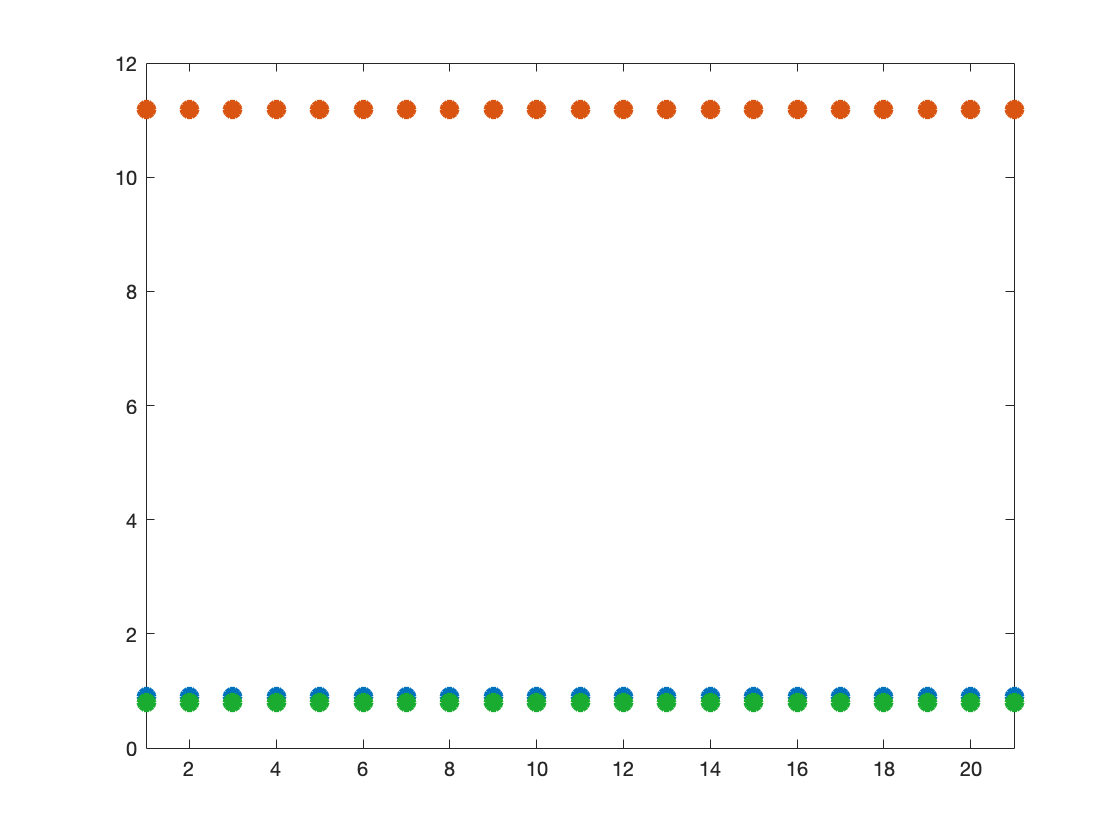}
    \subcaption{}
\end{subfigure}%
\begin{subfigure}[b]{0.23\textwidth}
    \includegraphics[width=\linewidth]{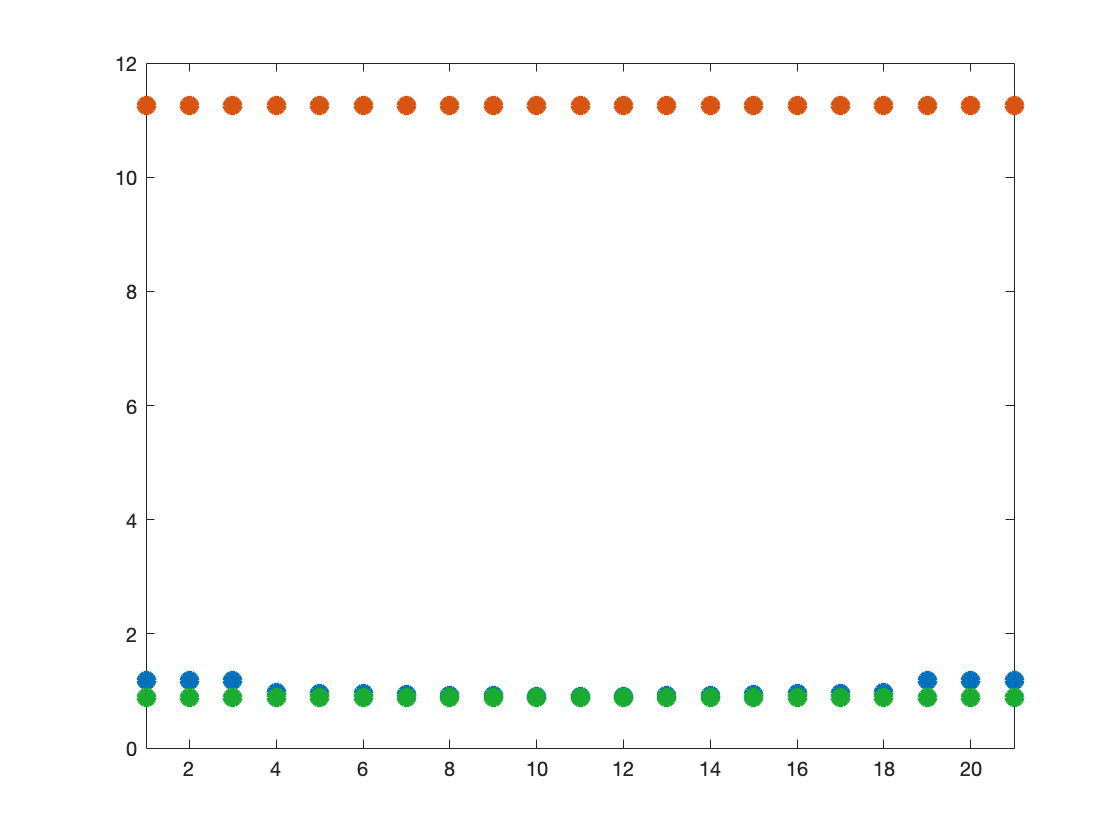}
    \subcaption{}
\end{subfigure}%
\begin{subfigure}[b]{0.23\textwidth}
    \includegraphics[width=\linewidth]{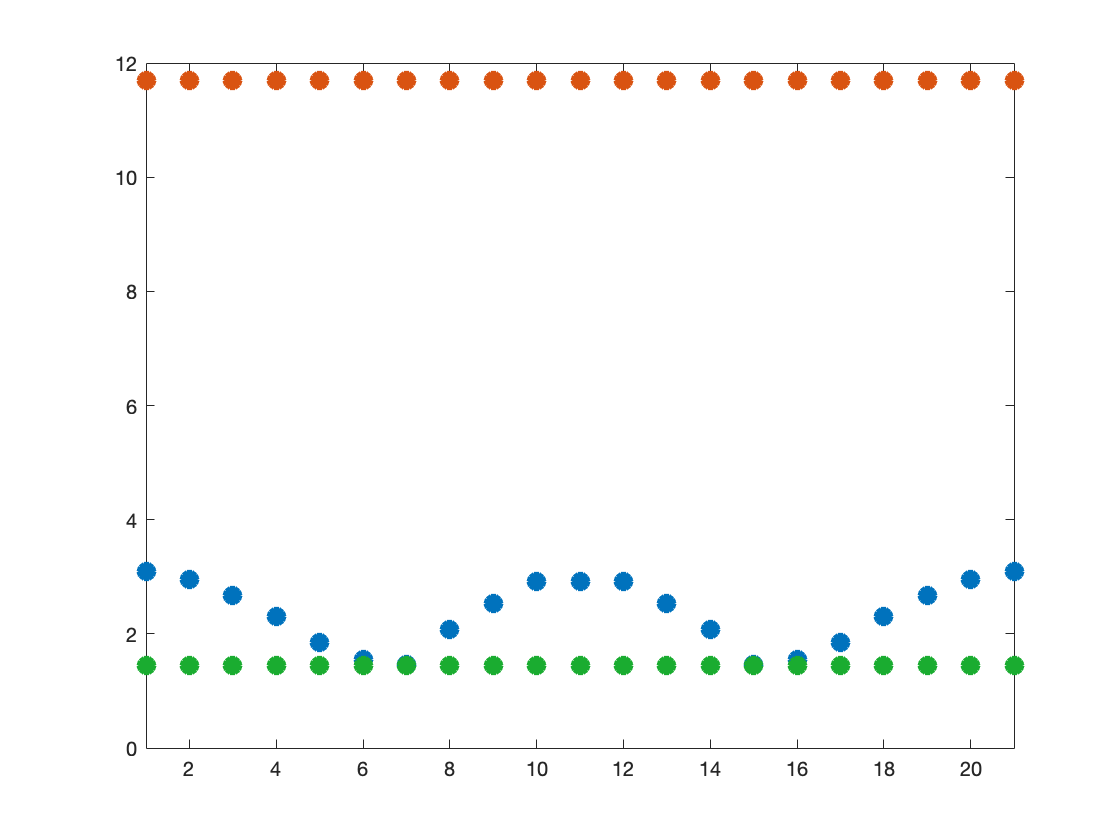}
    \subcaption{}
\end{subfigure}%
\begin{subfigure}[b]{0.23\textwidth}
    \includegraphics[width=\linewidth]{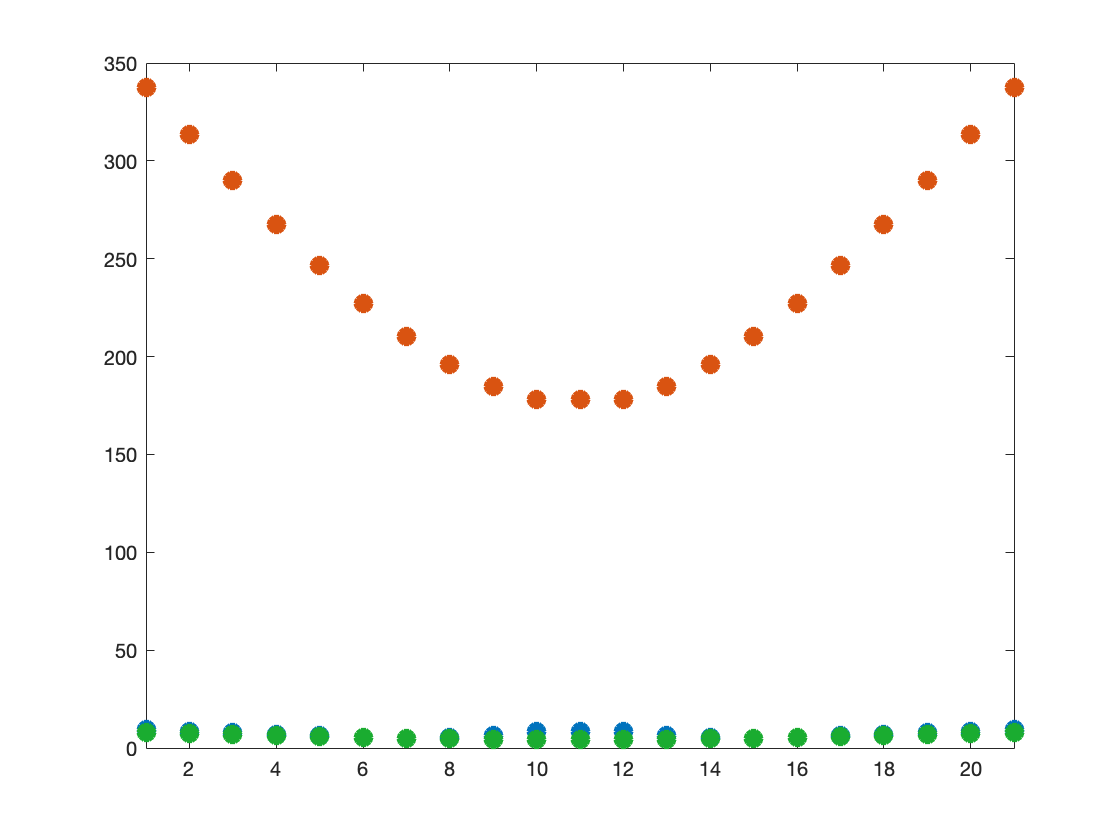}
    \subcaption{}
\end{subfigure}
\caption{
Distributionally Robust Cascading Risk under parameter uncertainty (rows) for various (equal edge weights) graph topologies (columns). 
The markers indicate the type of risk: 
\textcolor[rgb]{0.8500,0.3250,0.0680}{\textbullet} DR cascading risk, 
\textcolor[rgb]{0,0.4470,0.7410}{\textbullet} cascading risk without distributional robustness, 
\textcolor[rgb]{0.1000,0.6740,0.1880}{\textbullet} single-agent risk. 
Each row corresponds to a different level of parameter uncertainty, and each column corresponds to a graph topology (Complete, 14-cycle, 6-cycle, Path). 
}

\label{fig:dr_cascading_risk_profile}
\end{figure*}

\section{Special Graph Topologies}\label{sec:special_graphs}

In this section, we present specific results related to special unweighted graph topologies, which serve as a foundation for interpreting the outcomes in Section~\ref{sec:case_study}. In particular, we examine cases of symmetric graph structures where certain properties of the covariance matrix—induced by their graph Laplacians—exhibit special forms. These symmetric properties significantly influence the distributionally robust cascading risk profiles, as will be discussed later.

We begin with the case of the unweighted complete graph and formalize the key result in the form of the following lemma:

\begin{lemma}\label{lem:complete_graph_sigma_rho}
    For an unweighted complete graph, the diagonal entries and the correlation coefficients of the covariance matrix, as defined in Lemma~\ref{lem:sigma_y_steady}, are identical across all agents and are independent of their position in the graph.
\end{lemma}

Next, we consider the unweighted \( p \)-cycle graph and utilize the symmetry of its Laplacian matrix to derive the following result:

\begin{lemma}\label{lem:p_cycle_graph_sigma}
    For an unweighted \( p \)-cycle graph, all diagonal entries of the covariance matrix, as defined in Lemma~\ref{lem:sigma_y_steady}, are equal, i.e., \( \sigma_k = \sigma_l \) for all \( k, l \in \{1, \dots, n\} \).
\end{lemma}

These two lemmas highlight the structural uniformity introduced by the symmetry of certain graph Laplacians. As we demonstrate in Section~\ref{sec:case_study}, this uniformity leads to consistent patterns in the distributionally robust cascading risk across agents.

\section{Case Studies} \label{sec:case_study}
We consider a team of \( n = 21 \) agents, where the nominal diffusion coefficient in the consensus model~\eqref{eqn: network_dynamics} is set to \( b_0 = 4 \), and the parameter \(c\) is chosen as \( c = 0.1 \), in accordance with Definition~\ref{def:c-Consensus}. Throughout all simulation studies, we assume that the observable of agent \( i = 11 \) lies outside the interval \( |c| \), specifically within the set \( U_{\delta^i}, \) where \(\delta_i = 5\).

We explore several case studies of the rendezvous problem, examining consensus dynamics governed by~\eqref{eqn: network_dynamics} under three distinct communication topologies: complete, path, and \( p \)-cycle graphs, as defined in~\cite{van2010graph}.




\subsection{DR Risk of Cascading Failure}
For the quantification of the distributionally robust (DR) cascading risk profile, as shown in Fig.~\ref{fig:dr_cascading_risk_profile}, we set the nominal parameters to \( \tau_0 = 0.05 \) and \( \omega_0 = 0.5 \). We assume a \(5\%\) uncertainty in the nominal values of these parameters.

The DR risk profile of cascading failures in achieving \( c \)-consensus is evaluated using the closed-form expression from Theorem~\ref{thm:DR_risk}, applying the approximation in~\eqref{eqn:expectation_approx} across different unweighted communication graph topologies, where all edge weights are set to 0.5. Parameter uncertainties are considered individually, in accordance with the ambiguity set definition described in Remark~\ref{rem:ambiguity_set_quant_framework} and further detailed in Propositions \ref{prop:ambiguity_b}, \ref{prop:ambiguity_time_delay}, \ref{prop:ambiguity_edge_weight_tau_zero} and \ref{prop:ambiguity_edge_weight_tau_non_zero}.

Figure~\ref{fig:dr_cascading_risk_profile} presents the DR cascading risk profiles under uncertainty for each parameter separately, and compares them with the baseline scenarios: no parameter uncertainty and agent-level single-risk, which is calculated using Corollary \ref{cor:single_risk_rho_0}. In the following, we analyze the contribution of each uncertain parameter to the overall DR cascading risk in detail.

\subsubsection{Uncertainty in Diffusion Coefficient}
Figures~\ref{fig:dr_cascading_risk_profile}(a)--(d) illustrate the impact of uncertainty in the diffusion coefficient on the cascading risk profile. As noted in Remark~\ref{rem:opt_constraints_b}, uncertainty in the diffusion coefficient does not alter the correlation between agents \( j \) and \( i \).
For the complete unweighted graph, the risk profile remains uniform across all agents, as both the correlation coefficients and individual variances are identical for all agents, in accordance with Lemma~\ref{lem:complete_graph_sigma_rho}. In contrast, the results for the path and \( p \)-cycle graphs highlight the significant influence of network topology on the relative risk distribution among agents. Specifically, in the path graph and the 6-cycle graph, the risk is notably higher in the immediate and distant neighborhoods of the failed agent \( i \).



\subsubsection{Uncertainty in Time Delay}
Figures~\ref{fig:dr_cascading_risk_profile}(e)--(h) show the impact of uncertainty in the time-delay parameter on the cascading risk profile. As with the diffusion coefficient case, the uniformity of risk in the complete graph is explained by Lemma~\ref{lem:complete_graph_sigma_rho}. Interestingly, we also observe a uniform risk profile for the \( p \)-cycle graph. This is consistent with Lemma~\ref{lem:p_cycle_graph_sigma}, which states that the variances of all agents in a \( p \)-cycle graph are identical, combined with the constraint~\eqref{eqn:opt_constraint_rho}, which bounds the absolute value of correlation within the interval \([0, 1)\).

In contrast, the path graph exhibits a non-uniform risk distribution, where the distributionally robust cascading risk increases with distance from the failed agent \( i \) along the communication graph.


\subsubsection{Uncertainty in Network Weights}
Figures~\ref{fig:dr_cascading_risk_profile}(i)--(l) illustrate the impact of uncertainty in the network parameters on the cascading risk profile. The relative pattern of distributionally robust cascading risk is qualitatively similar to the case with uncertainty in the time-delay parameter. However, in the case of the path graph, we observe a significant amplification in the magnitude of the DR risk.

This behavior can be explained by examining the graph of the function \( f_{\tau}(\lambda_i) \) in Fig. \ref{fig:f_tau_lambda_graph}, where the product of the Laplacian eigenvalue and the time-delay lies close to the left asymptote of the function. In such a regime, even small variations in the edge weights of the graph can lead to large fluctuations in the system response, resulting in elevated cascading risk. Figure~\ref{fig:dr_risk_network_uncertainty_tau_0} illustrates the distributionally robust (DR) risk associated with fluctuations in network weights for the case of zero time delay. The empirical results demonstrate that the DR risk decreases as network connectivity increases. Furthermore, the figure underscores the role of parameter uncertainty in amplifying the risk of cascading large fluctuations.





\subsection{Impact of Network Connectivity on DR Risk}
While greater network connectivity is typically associated with enhanced robustness, it does {not} necessarily lead to a reduction in the distributionally robust cascading risk. We illustrate this by again considering the example of unweighted graphs with a uniform change in edge weights for simplicity. This effect is evident in Fig.~\ref{fig:dr_risk_vs_network_weights}, particularly for the complete graph and the $14$-cycle graph, where the risk displays a distinctly {non-monotonic} dependence on the edge weights. This behavior stems from the way the covariance matrix depends on the graph Laplacian, as characterized in Lemma~\ref{lem:sigma_y_steady}, which formalizes the non-monotonic relationship between connectivity and the covariance matrix within the cone of positive semidefinite matrices. Consequently, our results show that, in certain regions of the connectivity space (Fig.~\ref{fig:f_tau_lambda_graph}), increasing connectivity can, perhaps counterintuitively, {increase} the distributionally robust risk of cascading failures.

\section{Conclusion} \label{sec:conclusion}

This work develops a distributionally robust risk framework for analyzing cascading failures in time-delayed, first-order networked dynamical systems. We derive explicit formulas for the ambiguity sets of the probability distributions of observables under uncertainties in various model parameters, and provide an optimization framework for evaluating the distributionally robust cascading risk, offering both theoretical insights and empirical validation. Our results demonstrate how steady-state statistics influence risk in multi-agent rendezvous, highlighting the critical role of ambiguity sets and the nontrivial effect of network connectivity. Furthermore, we derive fundamental limits on the distributionally robust cascading risk as a function of network parameters and the radius of the ambiguity set, underscoring how connectivity and parameter uncertainty can exacerbate risk.  

Future work will focus on extending the framework to scenarios involving multiple agent failures, enabling resilience guarantees under more general fault patterns. These insights are expected to support the development of distributed control strategies that remain robust despite an arbitrary number of failed agents. In addition, we will pursue distributionally robust risk analysis under arbitrary variations in the weights of the underlying network, further enhancing the applicability of the proposed approach.

\begin{figure}[h]
    \centering
    \begin{subfigure}[t]{0.48\linewidth}
        \centering
        \includegraphics[width=\linewidth]{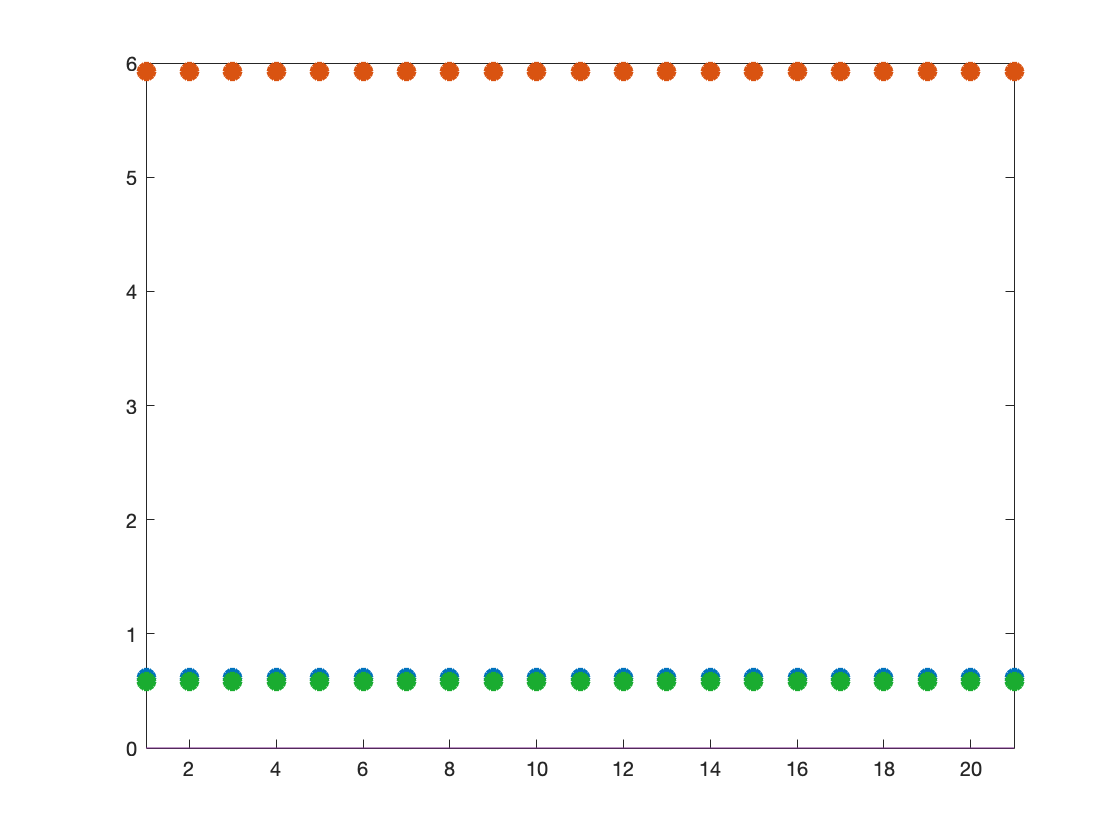}
        \caption{The complete graph.}
    \end{subfigure}
    \hfill
    \begin{subfigure}[t]{0.48\linewidth}
        \centering
        \includegraphics[width=\linewidth]{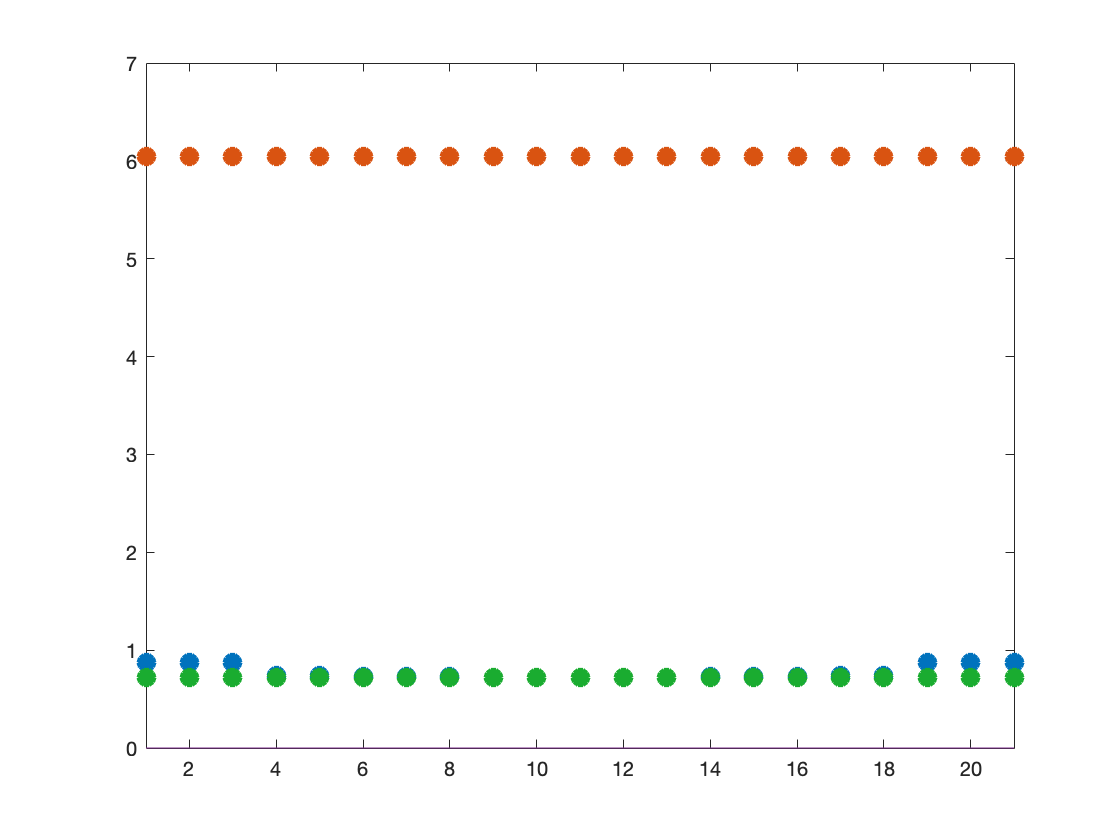}
        \caption{The $14$-cycle graph.}
    \end{subfigure}
    \vfill
    \begin{subfigure}[t]{0.48\linewidth}
        \centering
        \includegraphics[width=\linewidth]{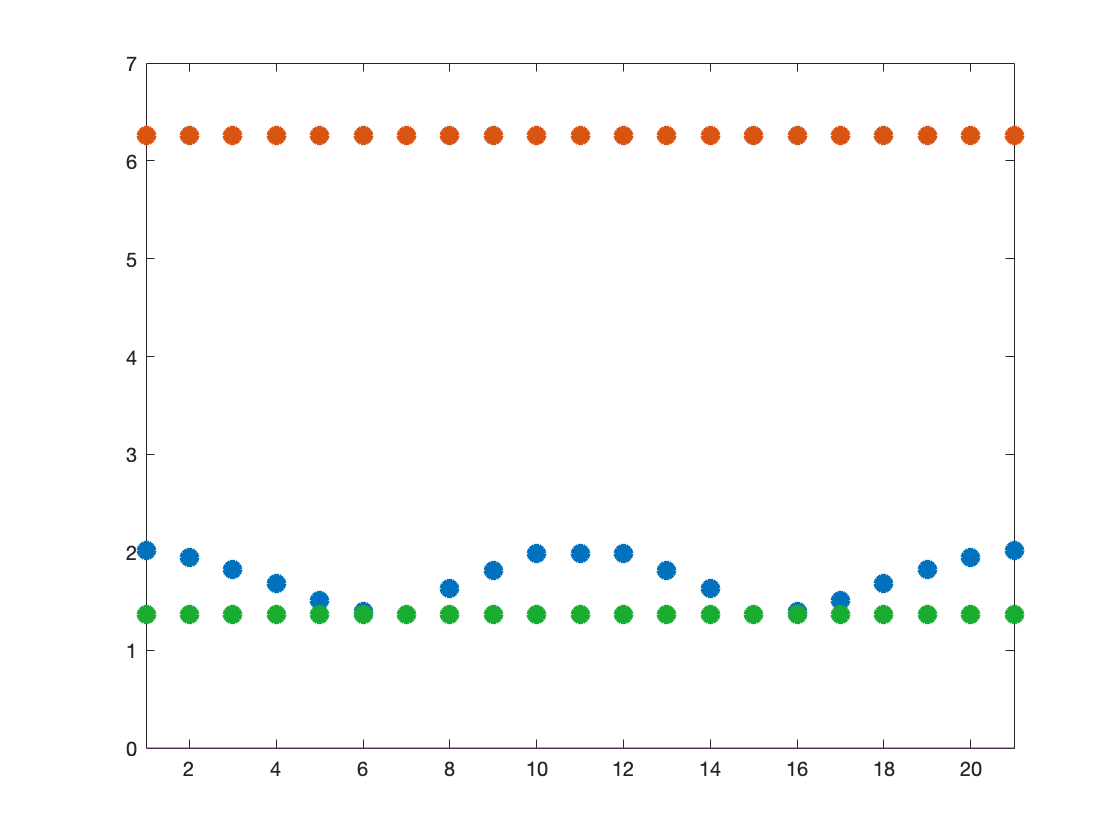}
        \caption{The $6$-cycle graph.}
    \end{subfigure}
    \hfill
    \begin{subfigure}[t]{0.48\linewidth}
        \centering
        \includegraphics[width=\linewidth]{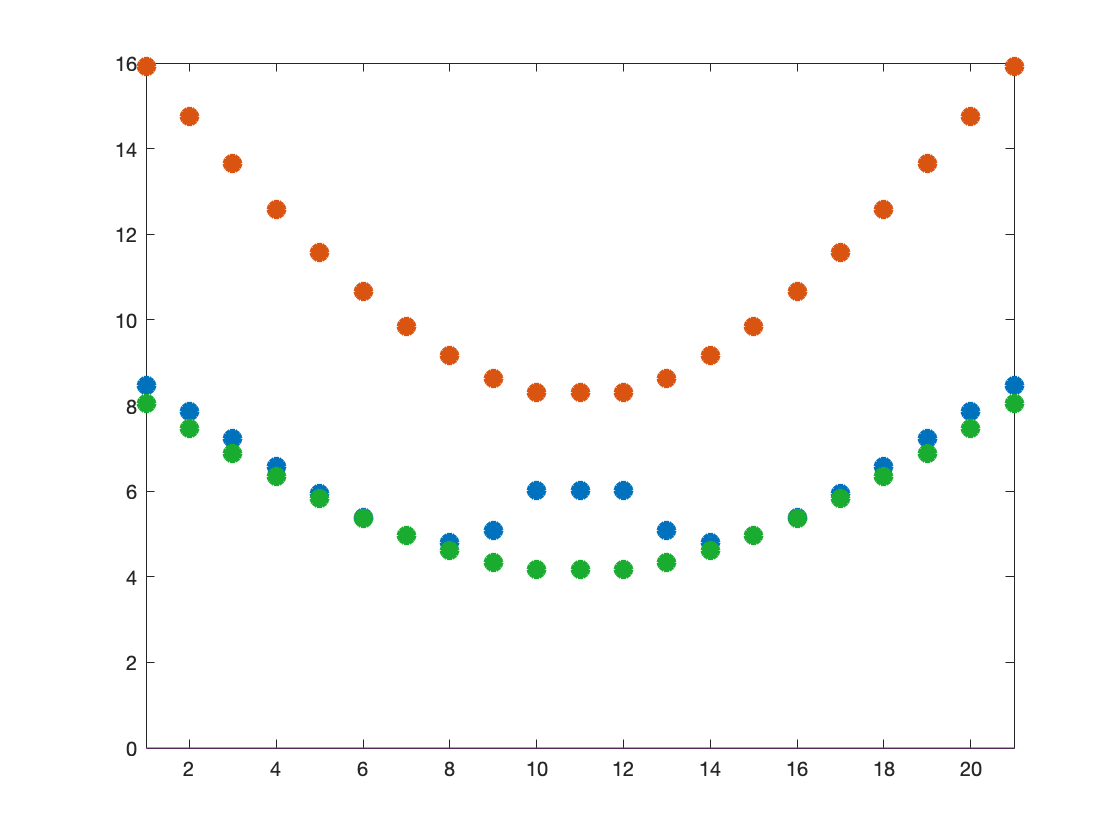}
        \caption{The path graph.}
    \end{subfigure}
    \caption{%
Distributionally robust cascading risk for uncertainty in network parameters with zero time delay. 
The markers indicate the type of risk: 
\textcolor[rgb]{0.8500,0.3250,0.0680}{\textbullet} DR cascading risk, 
\textcolor[rgb]{0,0.4470,0.7410}{\textbullet} cascading risk without distributional robustness, and 
\textcolor[rgb]{0.1000,0.6740,0.1880}{\textbullet} single-agent risk.
}
    \label{fig:dr_risk_network_uncertainty_tau_0}
\end{figure}

\begin{figure}[t]
    \centering
    \begin{subfigure}[t]{0.48\linewidth}
        \centering
        \includegraphics[width=\linewidth]{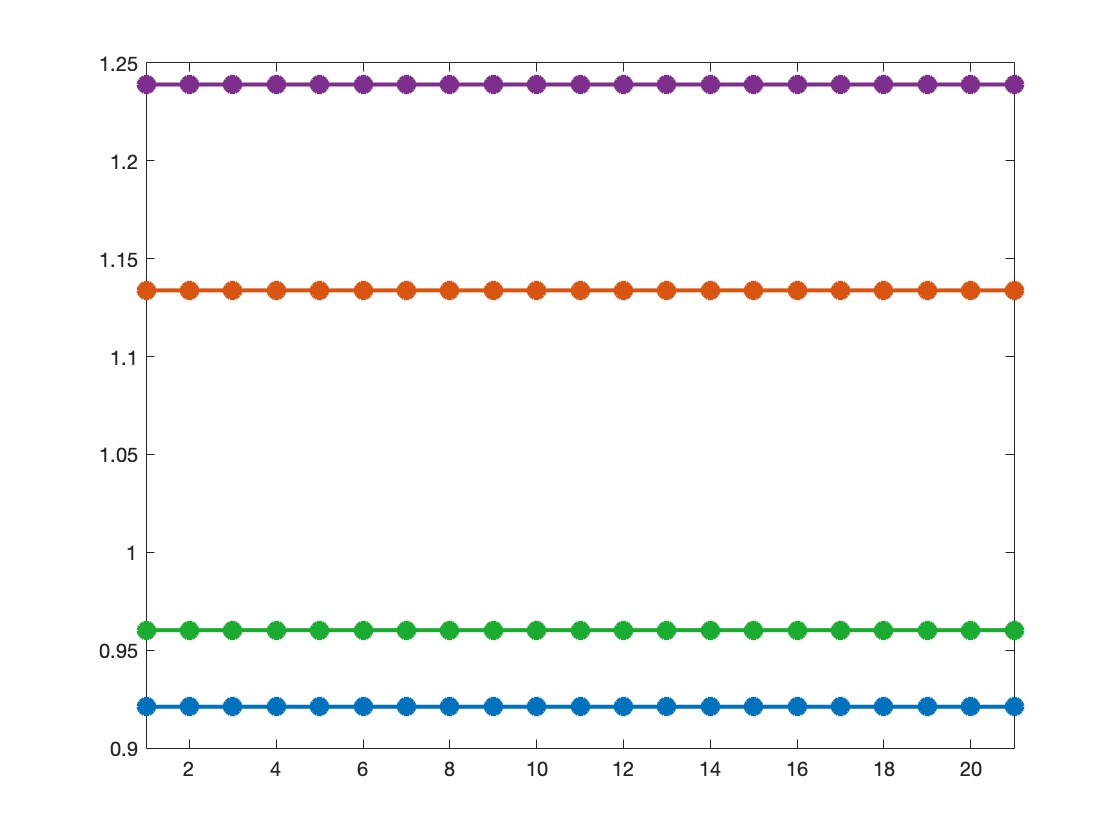}
        \caption{The complete graph.}
    \end{subfigure}
    \hfill
    \begin{subfigure}[t]{0.48\linewidth}
        \centering
        \includegraphics[width=\linewidth]{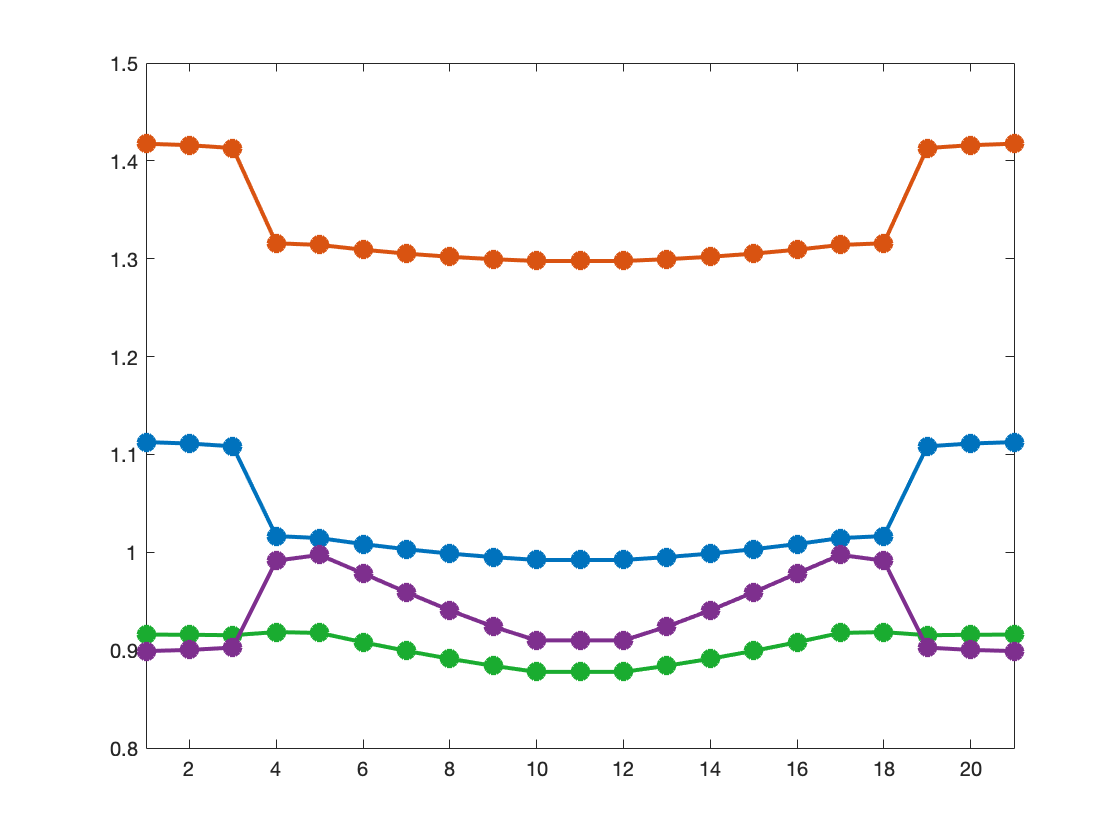}
        \caption{The $14$-cycle graph.}
    \end{subfigure}
    \vfill
    \begin{subfigure}[t]{0.48\linewidth}
        \centering
        \includegraphics[width=\linewidth]{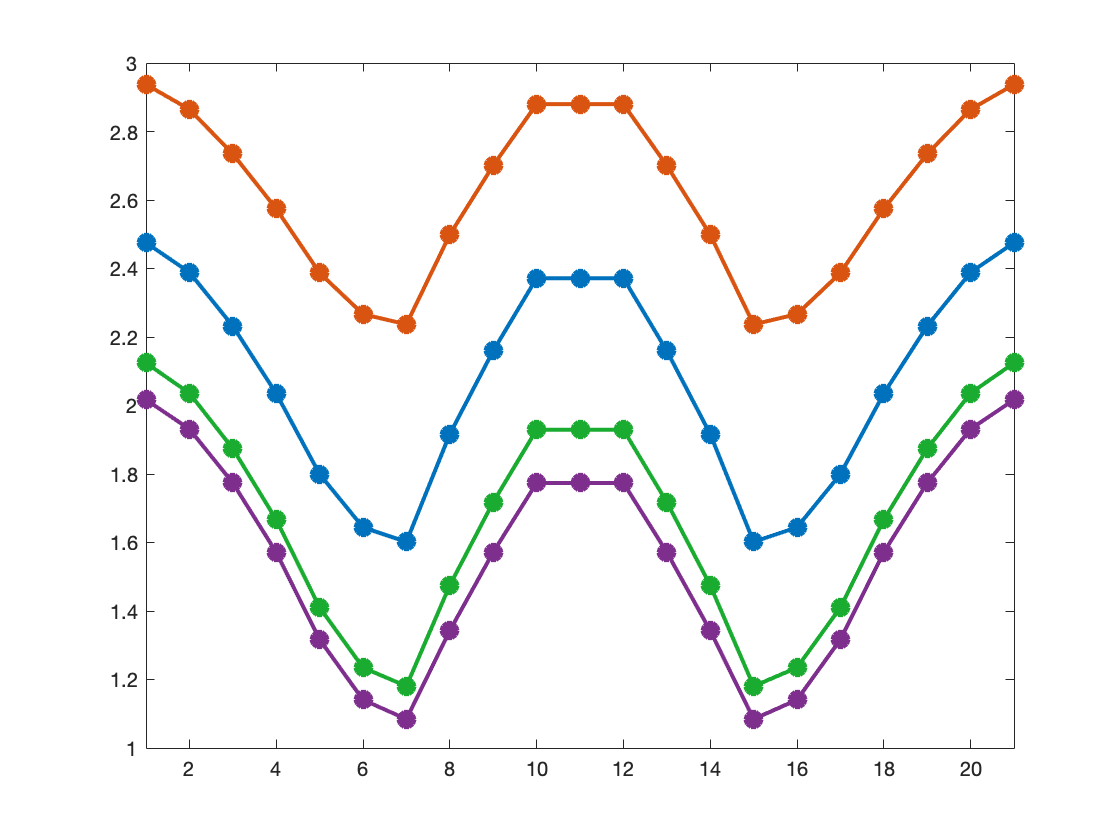}
        \caption{The $6$-cycle graph.}
    \end{subfigure}
    \hfill
    \begin{subfigure}[t]{0.48\linewidth}
        \centering
        \includegraphics[width=\linewidth]{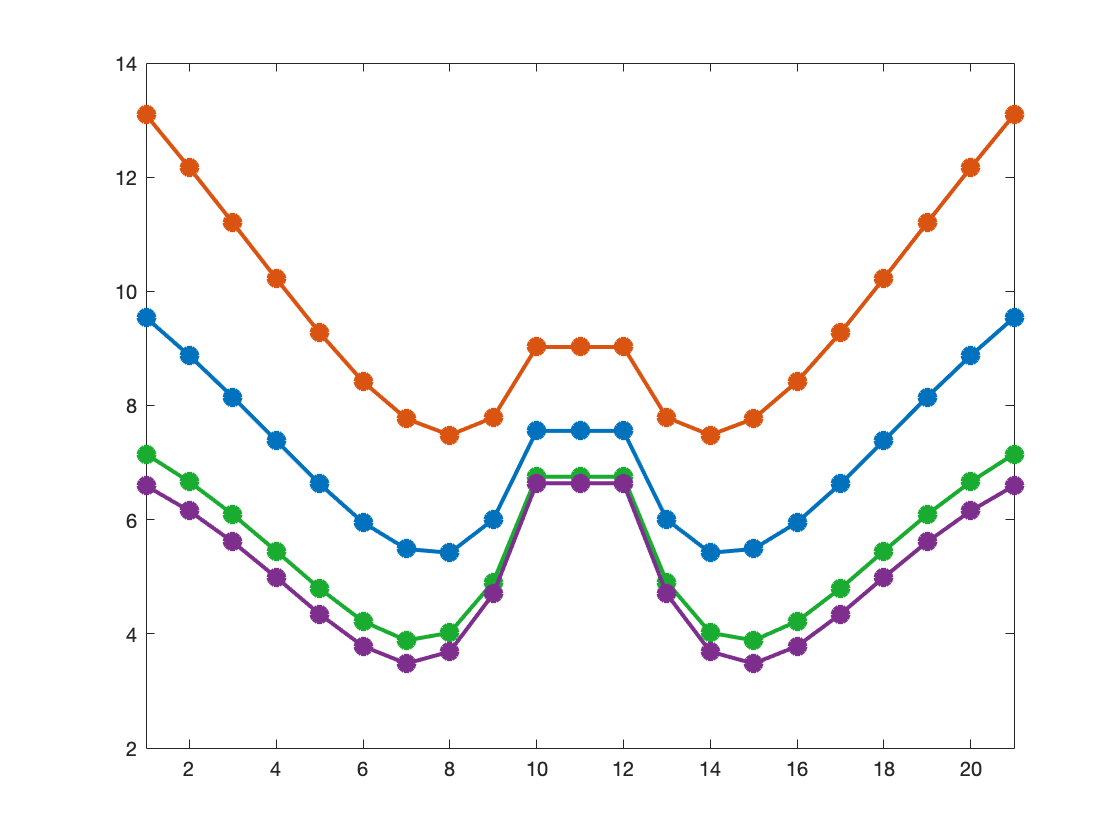}
        \caption{The path graph.}
    \end{subfigure}
    \caption{%
Distributionally robust cascading risk for uncertainty in diffusion coefficient with different weights of the communication graph. 
The markers indicate the different weights: 
\textcolor[rgb]{0.8500,0.3250,0.0680}{\textbullet} \(\omega = 0.25\), 
\textcolor[rgb]{0,0.4470,0.7410}{\textbullet} \(\omega = 0.5\), 
\textcolor[rgb]{0.1000,0.6740,0.1880}{\textbullet} \(\omega = 1\),
\textcolor[rgb]{0.4940,0.1840,0.5560}{\textbullet} \(\omega = 1.25\).
}
    \label{fig:dr_risk_vs_network_weights}
\end{figure}


\printbibliography
\begin{IEEEbiography}
[{\includegraphics[width=1in,height=1.25in,clip,keepaspectratio]{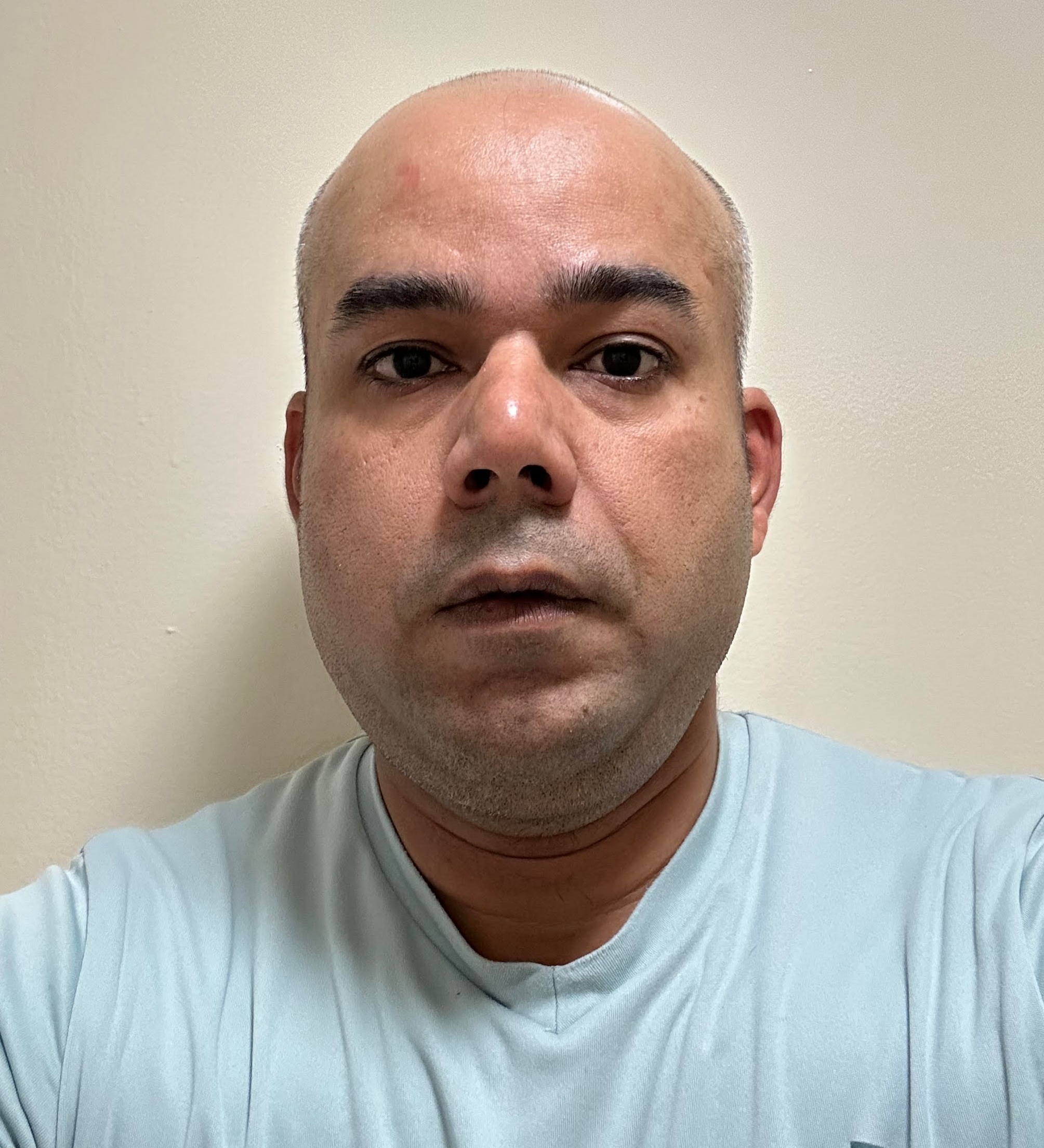}}]
{Vivek Pandey}
Vivek Pandey received his B.Tech and M.Tech degree in Chemical Engineering from Indian Institute of Technology, Mumbai, India in 2014. He is currently pursuing a Ph.D. degree in the Department of Mechanical Engineering and Mechanics at Lehigh University. His research interests include networked control systems.
\end{IEEEbiography}

\begin{IEEEbiography}[{\includegraphics[width=1in,height=1.25in,clip,keepaspectratio]{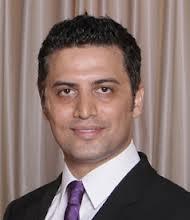}}]{Nader Motee}
Nader Motee (Senior Member, IEEE) received 
the B.Sc. degree in electrical engineering from 
the Sharif University of Technology, Tehran, 
Iran, in 2000, and the M.Sc. and Ph.D. degrees 
in electrical and systems engineering from the 
University of Pennsylvania, Philadelphia, PA, 
USA, in 2006 and 2007, respectively. 
From 2008 to 2011, he was a Postdoctoral 
Scholar with the Control and Dynamical Systems Department, California Institute of Technology, Pasadena, CA, USA. He is currently a 
Professor with the Department of Mechanical Engineering and Mechanics, Lehigh University, Bethlehem, PA, USA. His research interests
include distributed control systems and real-time robot perception. 
Dr. Motee was the recipient of several awards including the 2019 Best 
SIAM Journal of Control and Optimization Paper Prize, the 2008 AACC 
Hugo Schuck Best Paper Award, the 2007 ACC Best Student Paper 
Award, the 2008 Joseph and Rosaline Wolf Best Thesis Award, the 
2013 Air Force Office of Scientific Research Young Investigator Program 
Award, 2015 NSF Faculty Early Career Development Award, and a 2016 
Office of Naval Research Young Investigator Program Award.
 \end{IEEEbiography}

\appendix
\begin{proof}[Proof of Lemma \ref{lem:sigma_y_steady}]
   A detailed proof of this result is provided in \cite{Somarakis2019g}. 
\end{proof}


\begin{proof}[Proof of Lemma \ref{lem:principle_covariance_invertibility}]
    

We begin by proving that any \((n{-}1) \times (n{-}1)\) principal submatrix of \(\Sigma\) is invertible, and then extend the result to submatrices of smaller sizes.

From \eqref{eq:sigma_y}, it is evident that \(\Sigma\) is not invertible, with a nullspace spanned by the all-ones vector \(\bm{1}_n\).

Now, consider the principal submatrix \(\Sigma_{\vert i}\) obtained by removing the \(i\)th row and column of \(\Sigma\). We prove invertibility of \(\Sigma_{\vert i}\) by contradiction. Suppose \(\Sigma_{\vert i}\) is not invertible. Then there exists a nonzero vector \(\tilde{v} \in \mathbb{R}^{n-1}\) such that \(\tilde{v}^\top \Sigma_{\vert i} \tilde{v} = 0\).

Let \(P \in \mathbb{R}^{(n-1) \times n}\) be the matrix that selects all rows except the \(i\)th; specifically, \(P\) is formed by deleting the \(i\)th standard basis vector \(\bm{e}_i\). Then, it follows that
\begin{equation} \label{eqn:principle_submatrix_Sigma}
    \Sigma_{\vert i} = P \Sigma P^\top.
\end{equation}
This implies
\[
\tilde{v}^\top P \Sigma P^\top \tilde{v} = 0 \quad \Rightarrow \quad (P^\top \tilde{v})^\top \Sigma (P^\top \tilde{v}) = 0,
\]
so the vector \(P^\top \tilde{v} \in \mathbb{R}^n\) lies in the nullspace of \(\Sigma\). However, since \(P^\top \tilde{v}\) has zero in the \(i\)th coordinate by construction, it cannot be proportional to \(\bm{1}_n\), contradicting the assumption that the nullspace of \(\Sigma\) is spanned only by \(\bm{1}_n\).

Therefore, \(\Sigma_{\vert i}\) must be invertible for all \(i \in \{1, \dots, n\}\).

Finally, the invertibility of all smaller principal submatrices follows from the standard result that any principal submatrix of a positive definite matrix is itself positive definite (and thus invertible) \cite{horn2012matrix}.
\end{proof}

\begin{proof}[Proof of Lemma \ref{lem:bivariate_normal}]
Let $(y_i,y_j)^\top \sim \mathcal N(0,\Sigma_{\mid ij})$ with
\[
\Sigma_{\mid ij}=\begin{bmatrix}
\sigma_i^2 & \rho\,\sigma_i\sigma_j\\
\rho\,\sigma_i\sigma_j & \sigma_j^2
\end{bmatrix},\qquad \rho'=\sqrt{1-\rho^2}.
\]

By Lemma~\ref{lem:principle_covariance_invertibility}, $\Sigma_{\mid ij}$ is invertible, so the standard formula for the bivariate normal density applies \cite{tong2012multivariate}:
\[
f(y_j,y_i)
=\frac{1}{2\pi\rho'\sigma_i\sigma_j}
\exp\!\left(
-\frac{1}{2\rho'^2}\Big(\tfrac{y_i^2}{\sigma_i^2}-\tfrac{2\rho\,y_i y_j}{\sigma_i\sigma_j}+\tfrac{y_j^2}{\sigma_j^2}\Big)
\right).
\]
Completing the square in $y_i$ yields
\[
-\frac{1}{2\rho'^2}\Big(\tfrac{y_i^2}{\sigma_i^2}-\tfrac{2\rho\,y_i y_j}{\sigma_i\sigma_j}+\tfrac{y_j^2}{\sigma_j^2}\Big)
=-\frac{y_j^2}{2\sigma_j^2}
-\frac{\big(y_i-\rho\frac{\sigma_i}{\sigma_j}y_j\big)^2}{2\rho'^2\sigma_i^2},
\]
which gives the claimed expression. Equivalently, the joint density factors into the marginal 
$y_j \sim \mathcal N(0,\sigma_j^2)$ and the conditional 
$y_i \mid y_j \sim \mathcal N\!\left(\rho\frac{\sigma_i}{\sigma_j}y_j,\;\rho'^2\sigma_i^2\right)$.
\end{proof}




\begin{proof}[Proof of Proposition \ref{prop:ambiguity_b}]
We begin by constructing the ambiguity set for a general diffusion coefficient matrix \(B\) in \eqref{eqn: network_dynamics}.  
In the context of stochastic differential equations, the input noise covariance is given by \(\Pi = BB^\top.\)  
Suppose that \(\Pi\) admits the following bound:
\begin{equation}\label{eqn:B_uncertainty}
    (1-\alpha_B)\, \Pi_0 \;\preceq\; \Pi \;\preceq\; (1+\alpha_B)\, \Pi_0,
\end{equation}
where \(\Pi_0 = B_0 B_0^\top\) is the nominal input covariance matrix corresponding to the nominal diffusion parameter \(B_0.\)

From \cite{Somarakis2017a}, the steady-state covariance of \(\bm{y}_t\) is given by
\begin{align*}
    \Sigma 
    = \lim_{t \to \infty}\int_{0}^{t} 
      M_n \Phi_L(t-s)\, B B^\top \Phi_L^\top(t-s) M_n \, ds,
\end{align*}
where \(\Phi_L(\cdot)\) denotes the principal solution of the deterministic part of \eqref{eqn: network_dynamics}.  

Now consider the inequality 
\(\Pi = BB^\top \succeq (1-\varepsilon) \Pi_0.\)  
On the cone of positive semidefinite matrices, it follows that
\begin{equation} \label{eqn:ambiguity_set_cone_y}
    \Tilde{\Phi}_L(t-s)\, \Pi \,\Tilde{\Phi}_L^\top(t-s)  
    \;\succeq\; (1-\varepsilon)\, 
    \Tilde{\Phi}_L(t-s)\, \Pi_0 \,\Tilde{\Phi}_L^\top(t-s),
\end{equation}
where \(\Tilde{\Phi}_L(t-s) = M_n \Phi_L(t-s).\)  
Taking the integral and limit on both sides of \eqref{eqn:ambiguity_set_cone_y} establishes the desired lower bound.  

While the inequality \eqref{eqn:ambiguity_set_cone_y} applies to the general case, we restrict attention to the simplified setting \(B = b I_n\), and consider uncertainty only in the scalar parameter \(b\). Combined with Lemma~\ref{lem:sigma_y_steady}, this yields the stated result. The upper bound follows analogously.  
\end{proof}

\begin{proof}[Proof of Proposition~\ref{prop:ambiguity_time_delay}]
We begin by establishing the monotonicity of the function \( f_{\lambda_i}(\tau) \), as defined in \eqref{eqn:f_lambda_tau}, with respect to the delay parameter \( \tau \). Its derivative is given by:
\[
\frac{d}{d\tau} f_{\lambda_i}(\tau) = \frac{1}{\lambda_i \left(1 - \sin(\lambda_i \tau)\right)} > 0,
\]
for all \( i \in \{2, \dots, n\} \), where the positivity follows from Assumptions~\ref{asp:connected} and~\ref{asp:stable}, ensuring \( \lambda_i > 0 \) and \( \sin(\lambda_i \tau) < 1 \). Hence, \( f_{\lambda_i}(\tau) \) is strictly increasing in \( \tau \).

Now define the relative perturbation bound:
\[
\varepsilon_\tau^+ := \max_{i \in \{2, \dots, n\}} \frac{f_{\lambda_i}\left((1 + \alpha_\tau)\tau_0\right) - f_{\lambda_i}(\tau_0)}{f_{\lambda_i}(\tau_0)}.
\]
This implies the pointwise bound:
\[
f_{\lambda_i}\left((1 + \alpha_\tau)\tau_0\right) \leq (1 + \varepsilon_\tau^+) f_{\lambda_i}(\tau_0), \quad \forall i \in \{2, \dots, n\}.
\]

By spectral mapping, we then obtain the matrix inequality:
\[
\Psi\big|_{\tau = (1 + \alpha_\tau)\tau_0} \preceq (1 + \varepsilon_\tau^+) \Psi_0,
\]
where \( \Psi_0 := \Psi\big|_{\tau = \tau_0} \). Since congruence with \( Q \) and \( M_n \) preserves the Loewner order, it follows that:
\[
\Sigma = b_0 M_n Q \Psi Q^\top M_n^\top \preceq (1 + \varepsilon_\tau^+) \Sigma_0,
\]
where \( \Sigma_0 := b_0 M_n Q \Psi_0 Q^\top M_n^\top \).

An analogous argument for \( \tau = (1 - \alpha_\tau)\tau_0 \) yields a lower bound using:
\[
\varepsilon_\tau^- := \max_{i \in \{2, \dots, n\}} \frac{f_{\lambda_i}(\tau_0) - f_{\lambda_i}\left((1 - \alpha_\tau)\tau_0\right)}{f_{\lambda_i}(\tau_0)},
\]
which gives:
\[
\Sigma \succeq (1 - \varepsilon_\tau^-) \Sigma_0.
\]

Combining both bounds, we obtain:
\[
(1 - \varepsilon_\tau^-) \Sigma_0 \preceq \Sigma \preceq (1 + \varepsilon_\tau^+) \Sigma_0. \qedhere
\]
\end{proof}

\begin{proof}[Proof of Proposition~\ref{prop:ambiguity_edge_weight_tau_zero}]
Let the perturbed Laplacian be written as 
\[
L = B(W+W_{\Delta})B^\top = L_0 + \Delta.
\]
We first bound $L$ in terms of $L_0$. Observe that
\[
L \preceq (I + \Delta L_0^{\dagger})\,L_0,
\]
and since $\|I + \Delta L_0^{\dagger}\| = 1 + \|\Delta L_0^{\dagger}\|$, it follows that
\[
L \preceq \bigl(1 + \|\Delta L_0^{\dagger}\|\bigr)L_0.
\]
By a symmetric argument, the lower bound is obtained, yielding
\[
\bigl(1 - \|\Delta L_0^{\dagger}\|\bigr)L_0 
\;\preceq\; L 
\;\preceq\; \bigl(1 + \|\Delta L_0^{\dagger}\|\bigr)L_0.
\]

Since the steady-state covariance is given by $\Sigma = \tfrac{1}{2}L^{\dagger}$, we equivalently obtain
\begin{equation}\label{eqn:covariance_bound_inverse}
    \frac{1}{1 + \|\Delta L_0^{\dagger}\|}\,\Sigma_0 
    \;\preceq\; \Sigma 
    \;\preceq\; \frac{1}{1 - \|\Delta L_0^{\dagger}\|}\,\Sigma_0.
\end{equation}
Introducing the parameters
\[
\varepsilon_{\omega}^- := \frac{\|\Delta L_0^{\dagger}\|}{1+\|\Delta L_0^{\dagger}\|},
\qquad 
\varepsilon_{\omega}^+ := \frac{\|\Delta L_0^{\dagger}\|}{1-\|\Delta L_0^{\dagger}\|},
\]
the bounds in \eqref{eqn:covariance_bound_inverse} may be expressed in the form
\begin{equation}\label{eqn:covariance_bound_inverse_reversed}
    \bigl(1 - \varepsilon_{\omega}^-\bigr)\,\Sigma_0 
    \;\preceq\; \Sigma 
    \;\preceq\; \bigl(1 + \varepsilon_{\omega}^+\bigr)\,\Sigma_0.
\end{equation}

For sufficiently small perturbations, i.e., $\|\Delta L_0^{\dagger}\| \ll 1$, both $\varepsilon_\omega^\pm$ reduce to $\|\Delta L_0^{\dagger}\|$ to first order, so that the bound simplifies to
\[
\bigl(1 - \|\Delta L_0^{\dagger}\|\bigr)\Sigma_0 
\;\preceq\; \Sigma 
\;\preceq\; \bigl(1 + \|\Delta L_0^{\dagger}\|\bigr)\Sigma_0.
\]
The condition $\|\Delta L_0^{\dagger}\|\ll 1$ holds whenever $\|\Delta\|\ll \|L_0\|$.  
\end{proof}

\begin{proof}[Proof of Proposition ~\ref{prop:ambiguity_edge_weight_tau_non_zero}]
The proof proceeds along similar lines as the proof of Proposition~\ref{prop:ambiguity_time_delay}. Although this behavior is evident from Fig.~\ref{fig:f_tau_lambda_graph}, we first establish the existence of a unique minimum of \( f_{\tau}(\lambda) \) by solving the equation
\[
\frac{d}{d\lambda}\left(f_{\tau}(\lambda)\right) = 0,
\]
which leads to the condition
\[
\left\{\bar{\lambda} \;\middle|\; \lambda \tau - \cos(\lambda \tau) = 0 \right\}.
\]

Next, we consider two cases corresponding to the regions where the function \( f_{\tau}(\lambda_i) \), as shown in Fig.~\ref{fig:f_tau_lambda_graph}, is monotonic. We then follow the same steps as in the earlier proof: establishing a pointwise bound, deriving the corresponding matrix inequality, and finally quantifying the terms \( \varepsilon_{\omega}^{\pm} \).
\end{proof}


\begin{proof}[Proof of Lemma \ref{lem:conditional_expectation}]
    Using the notation in \cite{durrett2019probability}
\begin{align*}
    \mathbb{E}_{\mathbb{P}|\mathcal{G}_i}[\lvert y_j \rvert] &= \frac{\mathbb{E}_{\mathbb{P}|\mathcal{G}_i}[ \lvert y_j \rvert; U_{\delta^i}]}{\mathbb{P}\left(y_i \in U_{\delta^i} \right))}\\
    &= \dfrac{\mathbb{E}_{\mathbb{P}}[\lvert y_j \rvert\bm{1}_{y_i \in U_{\delta^i}}]}{\mathbb{P}[ {y_i \in U_{\delta^i}}]}.
\end{align*}
\end{proof}
\begin{proof}[Proof of Theorem \ref{thm:DR_risk}]

To prove this theorem, we state the following integration result as Lemma \ref{lem:korotkov_integral}.
\begin{lemma}\label{lem:korotkov_integral}
The following integral identities, adapted from \cite{korotkov2020_error}, will be used in the subsequent analysis:
\begin{multline} \label{eqn:korotkov_integral_two}
\int_{0}^{\infty} z \, \textnormal{erf}\left(a_1 z + b_1\right) \exp\left(-a_2 z^2\right) \, dz 
= \frac{\textnormal{erf}\left(b_1\right)}{2 a_2} \\
+ \frac{a_1}{2 a_2 \sqrt{A}} \exp\left(- \frac{a_2 b_1^2}{A}\right) \left(1 - \textnormal{erf}\left(\frac{a_1 b_1}{\sqrt{A}}\right)\right),
\end{multline}

\begin{equation} \label{eqn:korotkov_integral_one}
\int_{0}^{\infty} z \, \exp\left(-a_2^2 z^2\right) \, dz = \frac{1}{2 a_2^2}.
\end{equation}
where \(a_1 \geq 0\), \(a_2 > 0\), and \(A = a_1^2 + a_2\).
\end{lemma}
Using Lemma \ref{lem:conditional_expectation}
\[
    \mathbb{E}\left[ \lvert y_j \rvert| y_i \in U_{\delta^i}\right] =  \frac{\mathbb{E}\left[ \vert y_j \rvert \bm{1}_{\left[y_i \in U_{\delta^i}\right]}\right]}{\mathbb{P}\left[\bm{1}_{\left[y_i \in U_{\delta^i}\right]}\right]}
\]


We analyze the numerator and denominator separately.
\begin{align} \label{eqn: conditional_expectation_denominator}
\mathbb{P}\left[y_i \in U_{\delta^i}\right] \notag
    &= \int_{U_{\delta^i}}^{} ~ \frac{1}{\sigi \sqrt{2 \pi}} \text{exp}\left(- \frac{y_i^2}{\sigi^2}\right) dy_i \notag\\
    & = 1 - \text{erf} \left({\delta^{*}}\right),
\end{align}
where \(\delta^{*}\) is as defined in Theorem \ref{thm:DR_risk}.

Using joint distributions of \(y_j\) and \(y_i\) in Lemma \ref{lem:bivariate_normal}, 
\begin{align*}
    \mathbb{E}\left[ \lvert y_j \rvert \bm{1}_{\left[y_i \in U_{\delta^i}\right]}\right] &= \int_{-\infty}^{\infty} \int_{y_i \in U_{\delta^i}}^{} ~ \lvert y_j \rvert p(y_j,y_i)  dy_j dy_i\\
         &= \frac{2}{\sqrt{2 \pi} \sigj} \int_{0}^{\infty} y_j ~\text{exp}\left(-\frac{y_j^2}{2\sigj^2} \right) I_i dy_j,
\end{align*}
where 
\begin{align*}
    I_i &= \frac{1}{\sqrt{2 \pi} \rho' \sigi} \int_{y_i \in U_{\delta^i}}^{}  \text{exp}\left(- \frac{1}{2 \rho'^2 \sigi^2}  \left({y_i} - \rho \frac{y_j \sigi}{\sigj}\right)^2\right) dy_i,\\
    &= 1 - \frac{1}{2}\left(\text{erf}\left(\frac{\bar{\delta} + \rho \frac{\sigi}{\sigj}y_j}{\sqrt{2}\rho'\sigi}\right) + \text{erf}\left(\frac{\bar{\delta} - \rho \frac{\sigi}{\sigj}y_j}{\sqrt{2}\rho'\sigi}\right)\right),
\end{align*}

For simplicity of notation, consider  \(\kappa_{\bar{\delta}}^{\pm} = \frac{\bar{\delta} \pm \rho \frac{\sigi}{\sigj}y_j}{\sqrt{2}\rho'\sigi}\).

\begin{equation}\label{eqn:conditional_expectation_numerator}
        \mathbb{E}\left[\lvert y_j \rvert\bm{1}_{\left[y_i \in U_{\delta^i}\right]}\right] = \frac{1}{\sqrt{2 \pi} \sigj}\left(I_1 -  I_2\right), 
\end{equation}
where 
\begin{align*}
    I_1&= 2\int_{0}^{\infty} y_j  ~\text{exp}\left(-\frac{y_j^2}{2\sigj^2} \right) dy_j =   \mathbb{E}_{\mathbb{P}}[\lvert y_j \rvert],\\
    I_2 &=\int_{0}^{\infty} y_j \left(\text{erf}\left(\kappa_{\bar{\delta}}^{+}\right) + \text{erf}\left(\kappa_{\bar{\delta}}^{-}\right)\right)~\text{exp}\left(-\frac{y_j^2}{2\sigj^2} \right) dy_j
 \end{align*}

Using the result from Lemma \ref{lem:korotkov_integral} and substituting \(a_2 = \frac{1}{2 \sigj^2}\) in \eqref{eqn:korotkov_integral_one}, \(I_1\) is given by 
\begin{equation}\label{eqn:I1_evaluation}
    I_1 =   2\sigj^2 \hspace{0.5cm} \mathbb{E}_{\mathbb{P}}\left[\lvert y_j\rvert\right] = \sqrt{\frac{2}{\pi}} \sigma_j.
\end{equation}

Consider the following notation
\[I_2^{\pm} =\int_{0}^{\infty}\lvert y_j \rvert \left(\text{erf}\left(\kappa_{{\bar{\delta}}}^{\pm}\right) \right)~\text{exp}\left(-\frac{y_j^2}{2\sigj^2} \right) dy_j.\]

 Substituting the following in \eqref{eqn:korotkov_integral_one} to solve for \(I_2^{\pm}\):
  \[z = y_j \hspace{0.4cm} a_1 = \frac{\rho}{\sqrt{2} \rho' \sigj}, \hspace{0.4cm} b_1^{\pm} = \pm\frac{\bar{\delta}}{\sqrt{2}\rho' \sigi}, \hspace{0.4cm} a_2 = \frac{1}{2\sigj^2},\]

\begin{equation}\label{eqn:I2_evaluation}
    I_2 = \sqrt{\frac{2}{\pi}} \sigj \left[\textnormal{erf}\left(\frac{\delta^{*}}{\rho'}\right) -\rho \textnormal{erf}\left(\frac{\rho \delta^*}{\rho'}\right)\textnormal{exp}\left(-{\delta^{*^2}}\right)\right]
\end{equation}
The conditions in Lemma \ref{lem:korotkov_integral} is trivially satisfied for \(a_2\). For \(a_1\), without loss of generality \(\rho\) can be assumed positive for solving the integral in Lemma \ref{lem:korotkov_integral}, although in our problem \(\rho \in \left[-1,1\right].\) This is due to the fact that 
\[\text{erf}\left(\kappa_{\bar{\delta}}^{+}\left(-\rho\right)\right) = \text{erf}\left(\kappa_{\bar{\delta}}^{-}\left(\rho\right)\right) \]
The result follows from \eqref{eqn: conditional_expectation_denominator}, \eqref{eqn:conditional_expectation_numerator},\eqref{eqn:I1_evaluation}, \eqref{eqn:I2_evaluation} and \eqref{eqn:condn_dr_risk_definition}.
\end{proof}

\begin{proof}[Proof of Corollary \ref{cor:single_risk_rho_0}]
    From \eqref{eqn:I1_evaluation} in the proof of Theorem \ref{thm:DR_risk}, we obtain  
\[
\mathbb{E}_{\mathbb{P}}\left[\lvert y_j \rvert\right] = \sqrt{\frac{2}{\pi}} \sigma_j.
\]  
Substituting \(\rho = 0\) in Theorem \ref{thm:DR_risk} yields \(\rho' = 1\). The result follows directly by canceling the corresponding terms in the numerator and denominator and then taking the supremum over the ambiguity set defined in \eqref{eqn:ambiguity_set_b}.  
\end{proof}

\begin{proof}[Proof of Proposition~\ref{prop:opt_formulation_over_ambiguity_set}]
The objective function \eqref{eqn:opt_objective} follows directly from Theorem~\ref{thm:DR_risk}. The constraints on the marginal variances \(\sigma_i\) and \(\sigma_j\), given in \eqref{eqn:opt_constraint_sigi}–\eqref{eqn:opt_constraint_sigj}, are derived from the cone inequality defining the ambiguity set \(\mathcal{M}_p\). Specifically, for any covariance matrix \(\Sigma \in \mathcal{M}_p\) satisfying
\[
(1 - \varepsilon_p^-)\Sigma_0 \preceq \Sigma \preceq (1 + \varepsilon_p^+)\Sigma_0,
\]
it follows from the Loewner ordering that each marginal variance satisfies
\[
\sqrt{(1 - \varepsilon_p^-)}\,\sigma_{0,\star} \leq \sigma_{\star} \leq \sqrt{(1 + \varepsilon_p^+)}\,\sigma_{0,\star}, \quad \star \in \{i, j\}.
\]

To justify the constraint \eqref{eqn:opt_constraint_rho} on the correlation coefficient \(\rho\), we consider the structure of the off-diagonal entry \(\sigma_{ij}\), which, from Lemma~\ref{lem:sigma_y_steady}, can be written as
\[
\sigma_{ij} = \tilde{q}_i^\top\, \Psi\, \tilde{q}_j,
\]
where \(\tilde{q}_i\) is the \(i\)th column of the matrix \(M_n Q\), 

Using this structure, the correlation coefficient becomes
\[
\rho = \frac{\tilde{q}_i^\top \Psi \tilde{q}_j}{\sqrt{(\tilde{q}_i^\top \Psi \tilde{q}_i)(\tilde{q}_j^\top \Psi \tilde{q}_j)}}
= \frac{\sum_{k=2}^n \psi_k \tilde{q}_{ik} \tilde{q}_{jk}}{
\sqrt{
\left( \sum_{k=2}^n \psi_k \tilde{q}_{ik}^2 \right)
\left( \sum_{k=2}^n \psi_k \tilde{q}_{jk}^2 \right)
}}.
\]
Consider the special case of \(\mathcal{M}_p\), where \(\Sigma\) is simultaneously diagonalizable. Although the projection vectors \(\tilde{q}_i\) and \(\tilde{q}_j\) are fixed, the eigenvalues \(\{\psi_k\}\) vary within element-wise bounds due to the cone inequality. The resulting expression for \(\rho\) is thus a non-convex, nonlinear rational function of \(\{\psi_k\}\). Because the numerator and denominator respond differently to these uncertainties, no closed-form non-trivial bound on \(|\rho|\) can be established in general. This motivates the use of a conservative structural bound:
\[
|\rho| < 1,
\]
which ensures both feasibility and continuity of the objective function. 

Finally, since the objective \(\mathbb{E}_i^j\left[|y_j|\right]\) is an even function of \(\rho\), the constraint on \(\rho\) can be equivalently written as \(0 \leq |\rho| < 1\). The strict inequality \(|\rho| < 1\) is guaranteed by Lemma~\ref{lem:principle_covariance_invertibility}, which ensures positive definiteness of the marginal covariance matrices and thus non-degenerate normalization in the definition of \(\rho\).
\end{proof}

\begin{lemma}\label{lem:random_var_decomposition}
Let \(y_i\) and \(y_j\) be jointly Gaussian random variables with variances 
\(\sigma_i^2\) and \(\sigma_j^2\), respectively, and correlation coefficient \(\rho\).  
Then \(y_j\) can be decomposed as
\begin{equation}\label{eqn:random_variable_decomposition}
    y_j = \rho \frac{\sigma_j}{\sigma_i} \, y_i + z,
\end{equation}
where \(z \sim \mathcal{N}\!\left(0, \rho'^2 \sigma_j^2\right)\), \(\rho' = \sqrt{1-\rho^2}\), and  
\(z\) is independent of \(y_i\).
\end{lemma}



\begin{proof}[Proof of Lemma \ref{lem:random_var_decomposition}]
This result follows from the well-known fact that the sum of independent Gaussian random variables is Gaussian \cite{durrett2019probability}.  

Let \(z \sim \mathcal{N}(0, \rho'^2 \sigj^2)\) be independent of \(y_i \sim \mathcal{N}(0,\sigi^2)\). Consider the random variable
\[
\widetilde{y}_j := \rho \frac{\sigj}{\sigi} y_i + z.
\]

First, we check the expectation:
\[
\mathbb{E}[\widetilde{y}_j] 
= \rho \frac{\sigj}{\sigi} \, \mathbb{E}[y_i] + \mathbb{E}[z] 
= 0.
\]
Thus, \(\mathbb{E}[\widetilde{y}_j] = \mathbb{E}[y_j] = 0\).  

Next, we check the second moment:
\begin{align*}
\mathbb{E}[\widetilde{y}_j^2] 
&= \mathbb{E}\!\left[\left(\rho \frac{\sigj}{\sigi} y_i + z \right)^2\right] \\
&= \rho^2 \frac{\sigj^2}{\sigi^2} \, \mathbb{E}[y_i^2] + \mathbb{E}[z^2] \tag{independence, zero mean}\\
&= \rho^2 \frac{\sigj^2}{\sigi^2} \cdot \sigi^2 + \rho'^2 \sigj^2 \\
&= \sigj^2.
\end{align*}
Thus, \(\widetilde{y}_j \sim \mathcal{N}(0,\sigj^2)\).  

By the fact that the characteristic function uniquely determines the distribution \cite{durrett2019probability}, we conclude that 
\(
\widetilde{y}_j \overset{d}{=} y_j.
\)
\end{proof}





\begin{lemma}\label{lem:Sigma_rayleigh_bound}
Let $\Sigma \in \mathbb{R}^{n \times n}$ be the covariance matrix defined in Lemma~\ref{lem:sigma_y_steady}, and let $0 = \psi_1 < \psi_2 \le \dots \le \psi_n$ denote its eigenvalues, with $\psi_1$ corresponding to the trivial eigenvector $\mathbf{1}_n$. Define
\[
\tilde{\psi}_2 := \psi_2 \left(1 - \frac{1}{n}\right), \qquad \tilde{\psi}_n := \psi_n \left(1 - \frac{1}{n}\right).
\]

Then, every diagonal entry of $\Sigma$ satisfies
\[
\tilde{\psi}_2 \le \sigma_i^2 \le \tilde{\psi}_n, \qquad \text{for all } i = 1, \dots, n.
\]

Equivalently, the variance of each component is lower-bounded by the smallest nontrivial eigenvalue and upper-bounded by the largest eigenvalue of $\Sigma$, both scaled by the projection onto the subspace orthogonal to $\mathbf{1}_n$.
\end{lemma}

\begin{proof}
Decompose $\Sigma$ in its eigenbasis:
\[
\Sigma = \sum_{i=1}^n \psi_i \bar{q}_i \bar{q}_i^T = \sum_{i=2}^n \psi_i \bar{q}_i \bar{q}_i^T + 0 \cdot \bar{q}_1 \bar{q}_1^T,
\]
where $\bar{q}_1 = \mathbf{1}_n / \sqrt{n}$.

For any $i$,
\[
\sigma_i^2 = e_i^T \Sigma e_i = \sum_{k=2}^n \psi_k (e_i^T \bar{q}_k)^2.
\]

Since $\sum_{k=2}^n (e_i^T \bar{q}_k)^2 = 1 - \frac{1}{n}$, we can bound $\sigma_i^2$ by replacing all $\psi_k$ with either the smallest or largest eigenvalue in the sum:
\[
\sigma_i^2 \ge \psi_2 \sum_{k=2}^n (e_i^T \bar{q}_k)^2 = \psi_2 \left(1 - \frac{1}{n}\right) = \tilde{\psi}_2,
\]
\[
\sigma_i^2 \le \psi_n \sum_{k=2}^n (e_i^T \bar{q}_k)^2 = \psi_n \left(1 - \frac{1}{n}\right) = \tilde{\psi}_n.
\]

Hence, each diagonal entry of $\Sigma$ lies in the interval $[\tilde{\psi}_2, \tilde{\psi}_n]$.
\end{proof}


For the proofs of the next two theorems, we first establish the following auxiliary result:

\begin{lemma}\label{lem:conditional_expectation_mod_yi}
The conditional expectation of \( |y_i| \) given that \( y_i \in U_{\delta_i} \) admits the approximation
\begin{equation}\label{eqn:conditional_expecation_mod_yi}
    \mathbb{E}_{\mathbb{P}\mid \mathcal{F}^i}\!\left[\,|y_i|\,\right] 
    \;\approx\; 
    \frac{B\,\bar{\delta}_i}{1 - \exp\!\left(-A\,\frac{\bar{\delta}_i}{\sqrt{2}\,\sigma_i}\right)}\,.
\end{equation}
\end{lemma}

\begin{proof}[Proof of Lemma~\ref{lem:conditional_expectation_mod_yi}]
The result follows directly from the definition of conditional expectation. We have
\begin{align*}
    \mathbb{E}_{\mathbb{P}\mid \mathcal{F}^i}\!\left[\,|y_i|\,\right]
    &= \frac{\mathbb{E}_{\mathbb{P}}\!\left[\,|y_i|\,\mathbf{1}_{\{y_i \in U_{\bar{\delta}_i}\}}\right]}{\mathbb{P}\!\left(|y_i| \in U_{\bar{\delta}_i}\right)} \\[0.6em]
    &\overset{(a)}{=} 
    \frac{\mathbb{E}_{\mathbb{P}}\!\left[\,y_i\,\mathbf{1}_{\{y_i \geq \bar{\delta}_i\}}\right]}{\mathbb{P}\!\left(y_i \geq \bar{\delta}_i\right)} \\[0.6em]
    &\overset{(b)}{=} 
    \sqrt{\frac{2}{\pi}}\, \sigi
    \frac{\exp\!\bigl(-\tfrac{\bar{\delta}_i^2}{2\sigma_i^2}\bigr)}
    {\mathrm{erfc}\!\bigl(\tfrac{\bar{\delta}_i}{\sqrt{2}\,\sigma_i}\bigr)} \\[0.6em]
    &\overset{(c)}{\approx} 
    \frac{B\,\bar{\delta}_i}{1 - \exp\!\left(-A\,\frac{\bar{\delta}_i}{\sqrt{2}\,\sigma_i}\right)}\,.
\end{align*}
Here, \((a)\) follows from the symmetry of the Gaussian distribution (since \(y_i\) has zero mean), \((b)\) from evaluating the integral of the tail of a standard normal distribution, and \((c)\) from applying the approximation of the complementary error function given in~\eqref{eqn:erfc_approximation}.
\end{proof}




\begin{proof}[Proof of Theorem \ref{thm:dr_risk_upper_bound}]

To prove this theorem, we leverage the random variable decomposition from Lemma~\ref{lem:random_var_decomposition} and apply conditional expectation as follows:
\begin{align*}
    \mathbb{E}_{\mathbb{P}\mid \mathcal{F}^i}\!\left[\lvert y_j \rvert \right]
    &\overset{(a)}{=}\; \mathbb{E}_{\mathbb{P}\mid \mathcal{F}^i}\!\left[\;\bigl\lvert \rho \tfrac{\sigma_j}{\sigma_i} y_i + z \bigr\rvert \;\right] \\[0.3em]
    &\overset{(b)}{\leq}\; \mathbb{E}_{\mathbb{P}\mid \mathcal{F}^i}\!\left[\,\lvert \rho \rvert \tfrac{\sigma_j}{\sigma_i} \lvert y_i \rvert + \lvert z \rvert \,\right] \\[0.3em]
    &\overset{(c)}{=}\; \lvert \rho \rvert \tfrac{\sigma_j}{\sigma_i} \,\mathbb{E}_{\mathbb{P}\mid \mathcal{F}^i}\!\left[\lvert y_i \rvert \right] 
      + \mathbb{E}_{\mathbb{P}}\!\left[\lvert z \rvert \right] \\[0.3em]
    &\overset{(d)}{=}\; \sigma_j \Biggl( \lvert \rho \rvert\,\kappa\!\left({\sigma_i}\right) 
        + \sqrt{\tfrac{2}{\pi}}\, \rho' \Biggr),
\end{align*}
where  
(a) follows from Lemma~\ref{lem:random_var_decomposition},  
(b) from the triangle inequality~\cite{durrett2019probability},  
(c) from linearity of conditional expectation and independence of \(z\) (Lemma~\ref{lem:random_var_decomposition}), and  (d) from Lemma \ref{lem:conditional_expectation_mod_yi},  where \(\kappa\!\left(\cdot\right)\) as defined in \eqref{eqn:kappa_x}.

Maximizing over all admissible probability measures gives 
\begin{align*}
\sup_{\mathbb{P} \in \mathcal{M}_p} 
\mathbb{E}_{\mathbb{P} \mid \mathcal{F}_i}\!\left[\lvert y_j \rvert \right] 
\;&\leq\; 
\sup_{\mathbb{P} \in \mathcal{M}_p}
\sigma_j \Biggl( \lvert \rho \rvert\,\kappa\!\left({\sigma_i}\right) 
        + \sqrt{\tfrac{2}{\pi}}\, \rho' \Biggr).
\end{align*}

Using the feasibility constraints \eqref{eqn:opt_constraint_sigi}--\eqref{eqn:opt_constraint_rho}, we obtain (for \(\mathcal{M}_p = \{\mathcal{M}_{\tau}, \mathcal{M}_{\omega}\}, \))
\[
\mathcal{R}_i^j \leq \sup_{0\leq |\rho|<1}
\sigma_{0,j}^+ \Biggl( \lvert \rho \rvert\,\kappa\!\left({\sigma_{0,i}^+}\right) 
          + \sqrt{\tfrac{2}{\pi}}\, \rho' \Biggr),\]
\noindent where \(\sigma_{0,j}^+,\) and \(\sigma_{0,i}^+\) as defined in \eqref{eqn: sigma_+_-}. 

Invoking Lemma \ref{lem:Sigma_rayleigh_bound} yields
\[
\underset{\underset{i,j \in \{1, \dots, n\}}{i \neq j}}{\sup}\mathcal{R}_i^j \leq \sup_{0\leq |\rho|<1}
\sqrt{\tilde{\psi}_{0,n}^+} \Biggl( \lvert \rho \rvert\,\kappa\!\left({\sqrt{\tilde{\psi}_{0,n}^+}}\right) 
        + \sqrt{\tfrac{2}{\pi}}\, \rho' \Biggr),\]
\noindent where \(\tilde{\psi}_{0,n}^+,\) is as defined in \eqref{eqn: psi_+_-}. For notational convenience, we define \(
\kappa_i = \kappa\!\left({{\sqrt{\tilde{\psi}_{0,n}^+}}}\right).
\)

Finally, using basic trigonometric identities, we obtain
 \[
\underset{\underset{i,j \in \{1, \dots, n\}}{i \neq j}}{\sup}\mathcal{R}_i^j \leq \sup_{0\leq |\rho|<1}
\sqrt{\tilde{\psi}_{0,n}^+} \sqrt{\kappa_i^2 + \frac{2}{\pi} } \sin \left(\alpha + \beta\right) ,\]
where 
\[\tan \alpha = \frac{\rho'}{\rho} \qquad \tan \beta = \frac{\sqrt{2}}{\kappa_i \sqrt{\pi}}. \]
As a consequence, 
 \[
\underset{\underset{i,j \in \{1, \dots, n\}}{i \neq j}}{\sup}\mathcal{R}_i^j \leq 
 \sqrt{\left(\kappa_i^2 + \frac{2}{\pi} \right) \tilde{\psi}_{0,n}^+} .\]
For \(\mathcal{M}_{p} = \mathcal{M}_{b},\) the same analysis follows except \(\rho = \rho_0\) and
\[
\underset{\underset{i,j \in \{1, \dots, n\}}{i \neq j}}{\sup}\mathcal{R}_i^j \leq
\sqrt{\tilde{\psi}_{0,n}^+} \Biggl( \lvert \rho_0 \rvert\,\kappa\!\left({\sqrt{\tilde{\psi}_{0,n}^+}}\right) 
        + \sqrt{\tfrac{2}{\pi}}\, \rho'_0 \Biggr),\]
        The result then follows by \eqref{eqn:condn_dr_risk_definition}. 
\end{proof}

\begin{proof}[Proof of Theorem~\ref{thm:dr_risk_lower_bound}]
We derive a lower bound on the DR risk as follows:
\begin{align*}
    \mathbb{E}_{\mathbb{P} \mid \mathcal{F}_i}\!\left[\lvert y_j \rvert \right] 
    &\overset{(a)}{=} \mathbb{E}_{\mathbb{P} \mid \mathcal{F}_i}\!\left[\lvert \rho \frac{\sigma_j}{\sigma_i} y_i + z \rvert \right] \\[0.2em]
    &\overset{(b)}{\geq} \mathbb{E}_{\mathbb{P} \mid \mathcal{F}_i}\!\left[\;\bigl\lvert \lvert \rho \frac{\sigma_j}{\sigma_i} y_i \rvert - \lvert z \rvert \bigr\rvert \right] \\[0.2em]
    &\overset{(c)}{\geq} \biggl\lvert \mathbb{E}_{\mathbb{P} \mid \mathcal{F}_i}\!\left[\lvert \rho \frac{\sigma_j}{\sigma_i} y_i \rvert - \lvert z \rvert \right]\biggr\rvert \\[0.2em]
    &\overset{(d)}{=} \biggl\lvert \lvert \rho \frac{\sigma_j}{\sigma_i} \rvert\,\mathbb{E}_{\mathbb{P} \mid \mathcal{F}_i}\!\left[\lvert y_i \rvert \right] - \mathbb{E}_{\mathbb{P}}\!\left[\lvert z \rvert \right] \biggr\rvert \\[0.2em]
    &\overset{(e)}{=} \sigma_j \left\lvert \lvert \rho \rvert\, \kappa\!\left({\sigma_i}\right) - \sqrt{\frac{2}{\pi}}\, \rho' \right\rvert,
\end{align*}
where  
(a) follows from Lemma~\ref{lem:random_var_decomposition},  
(b) from the reverse triangle inequality~\cite{durrett2019probability},  
(c) from Jensen's inequality~\cite{durrett2019probability},  
(d) from linearity of conditional expectation and independence of \(z\) from Lemma~\ref{lem:random_var_decomposition}, and  \((e)\) follows from Lemma \ref{lem:conditional_expectation_mod_yi}.


Maximizing over admissible probability measures gives
\[
\sup_{\mathbb{P} \in \mathcal{M}_p} 
\mathbb{E}_{\mathbb{P} \mid \mathcal{F}_i}\!\left[\lvert y_j \rvert \right] 
\;\geq\; 
\sup_{\mathbb{P} \in \mathcal{M}_p}
\sigma_j\, \max\!\left\{ \sqrt{\frac{2}{\pi}}, \;\kappa\!\left({\sigma_i}\right) - \eta \right\},
\]
where $\eta > 0$ is arbitrarily small as $\lvert \rho \rvert \to 1$.  

Using the feasibility constraints \eqref{eqn:opt_constraint_sigi}--\eqref{eqn:opt_constraint_rho}, we obtain (for \(\mathcal{M}_p = \{\mathcal{M}_{\tau}, \mathcal{M}_{\omega}\} \))
\[
\mathcal{R}^j_i 
\;\geq\; 
\sigma_{0,j}^+\,
\max\!\left\{ \sqrt{\frac{2}{\pi}}, \;\kappa\!\left({\sigma_{0,i}^+}\right) - \eta \right\},
\]
\noindent where \(\sigma_{0,j}^+,\) and \(\sigma_{0,i}^+\) as defined in \eqref{eqn: sigma_+_-}. 

Finally, invoking Lemma~\ref{lem:Sigma_rayleigh_bound} yields
\[
\underset{\underset{i,j \in \{1, \dots, n\}}{i \neq j}}{\sup}\mathcal{R}^j_i 
\;\geq\; 
\sqrt{\tilde{\psi}_{0,2}^+}\,
\max\!\left\{ \sqrt{\frac{2}{\pi}}, \;\kappa\!\left({\sqrt{\tilde{\psi}_{0,2}^+}}\right) - \eta \right\},
\]
\noindent where \(\tilde{\psi}_{0,2}^+,\) is as defined in \eqref{eqn: psi_+_-}. 
For \(\mathcal{M}_{p} = \mathcal{M}_{b},\) the same analysis follows except \(\rho = \rho_0\) and
\[
\underset{\underset{i,j \in \{1, \dots, n\}}{i \neq j}}{\sup}\mathcal{R}_i^j \geq
\sqrt{\tilde{\psi}_{0,2}^+} \Biggl( \lvert \rho_0 \rvert\,\kappa\!\left({\sqrt{\tilde{\psi}_{0,2}^+}}\right) 
        - \sqrt{\tfrac{2}{\pi}}\, \rho'_0 \Biggr),\]
The result then follows by \eqref{eqn:condn_dr_risk_definition}. 
\end{proof}
\begin{proof}[Proof of Corollary \ref{cor:single_agent_upper_lower_bound}]
    The result follows trivially by substituting \(\rho = 0\) into Theorems \ref{thm:dr_risk_upper_bound} and \ref{thm:dr_risk_lower_bound}.
\end{proof}
\begin{proof}[Proof of Theorem \ref{thm:fundamental_limit}]
The proof follows directly by analyzing the dependence of \(\Psi\) on the eigenvalues of the Laplacian matrix, as elucidated in Lemma~\ref{lem:sigma_y_steady} and Remark~\ref{rem: Sigma_y_time_delay_zero}. The desired expression is obtained by substituting the maximum and minimum nontrivial eigenvalues of the covariance matrix into Theorem~\ref{thm:dr_risk_lower_bound}.
\end{proof}

\begin{proof}[Proof of Corollary \ref{cor:fundamental_limit_single_agent}]
The proof follows along similar lines as the proof of Theorem~\ref{thm:fundamental_limit}, leveraging the result of Corollary~\ref{cor:single_agent_upper_lower_bound}.
\end{proof}

\begin{proof}[Proof of Lemma \ref{lem:complete_graph_sigma_rho}]
    We start the proof by noting that the eigenvalues of unweighted complete graph are \(\lambda_1 = 0\) and \(\lambda_i = n\) for all \(i \in \{2, \dots, n\}.\) Then for all \(i \in \{2,\dots,n\},\) \(f(\lambda_i \tau) = f(n\tau).\)
    Then \(\Sigma\) as in Lemma \ref{lem:sigma_y_steady} can be written as 
        \begin{align*}
        \sigma_{ij} &= \tau f(n\tau) b^2 e_i^\top \left(\sum_{k = 2}^{n}  q_k q_k^\top \right) e_j \\
        &= \tau f(n\tau) b^2 e_i^\top \left(\sum_{k = 1}^{n}  q_k q_k^\top - q_1q_1^\top\right) e_j \\
                &= \tau f(n\tau) b^2 e_i^\top  \left(I_n - q_1q_1^\top \right) e_j 
    \end{align*}
    The \(kth\) diagonal element of \(\Sigma\) and correlation \(\rho\) can be written as 
    \[\sigma_k^2 = b {\tau f(n\tau)} \left(1 - \frac{1}{n}\right), \hspace{0.5cm} \rho = \frac{-1}{n-1}\]
\end{proof}

\begin{proof}[Proof of Lemma \ref{lem:p_cycle_graph_sigma}]
    We start the proof by noting that the Laplacian matrix of unweighted p-cycle graph is a circulant matrix, which is diagonalizable using Fourier basis \cite{van2010graph}. As a result, the eigenvector matrix is given as follows \cite{van2010graph}: 
    \[Q_{ij} = \frac{\omega^{ij}}{\sqrt{n}}, \quad \text{where} \quad \omega = e^{2\pi \iota/n}\] 
    Then \(\Sigma\) as in Lemma \ref{lem:sigma_y_steady} can be written as 
    \begin{align*}
              \sigma_{ij} &= \tau  b^2 e_i^\top \left(\sum_{k = 2}^{n} f(\lambda_k\tau) q_k q_k^H \right) e_j \\
        \sigma_{i}^2&= \tau  b^2  \left(\sum_{k = 2}^{n} f(\lambda_k\tau) q_{ki} q_{ki}^H \right) \\
        &=  \frac{\tau  b^2}{n}  \left(\sum_{k = 2}^{n} f(\lambda_k\tau) \right) 
    \end{align*}
\end{proof}

\end{document}